\documentclass[journal,12pt]{IEEEtran}
\normalsize
\usepackage{cite}
\usepackage{graphicx}
\usepackage{ifpdf}
\ifpdf
  \usepackage{epstopdf}
\fi
\usepackage[cmex10]{amsmath}
\usepackage{amsthm}
\usepackage{amssymb}
\usepackage{cases}
\usepackage{bm}
\usepackage[caption=false,font=footnotesize]{subfig}
\usepackage{fixltx2e}
\usepackage{color}
\usepackage{soul}

\begin{document}
%
\title{Statistical Multiplexing Gain Analysis of Heterogeneous Virtual Base Station Pools in Cloud Radio Access Networks}
%
%
%

\author{Jingchu Liu, ~\IEEEmembership{Student Member,~IEEE},
        Sheng Zhou, ~\IEEEmembership{Member,~IEEE},\\
        Jie Gong, ~\IEEEmembership{Member,~IEEE},
        Zhisheng Niu, ~\IEEEmembership{Fellow,~IEEE},\\
        and Shugong Xu, ~\IEEEmembership{Fellow,~IEEE}
\thanks{J. Liu, S. Zhou, and Z. Niu are with Department of Electronic Engineering and Tsinghua National Laboratory for Information Science and Technology, Tsinghua University, Beijing 100084, China. Email:
        liu-jc12@mails.tsinghua.edu.cn, sheng.zhou@tsinghua.edu.cn, niuzhs@tsinghua.edu.cn.}
\thanks{J. Gong was with Department of Electronic Engineering and Tsinghua National Laboratory for Information Science and Technology, Tsinghua University, Beijing 100084, China. He is now with School of Data and Computer Science, Sun Yat-sen University, Guangzhou 510006, Guangdong, China. Email: gongj26@mail.sysu.edu.cn.}
\thanks{S. Xu is with Intel Cooperation. Email: shugong.xu@intel.com.}
\thanks{Part of this work has been presented at IEEE Globecom 2014. This  work  is  sponsored  in  part  by  the  {National Basic  Research  Program  of China  (973 Program: 2012CB316001}), the National Science Foundation of China (NSFC) under grant No. 61201191, No. 61321061, No. 61401250, and No. 61461136004, and {Intel Collaborative Research Institute for Mobile Networking and Computing}.}}

%
%



\maketitle
\theoremstyle{plain}
\newtheorem{kcr}{Theorem}
\newtheorem{rev}[kcr]{Theorem}
\newtheorem{complex}[kcr]{Theorem}
\newtheorem{lpl}[kcr]{Theorem}
\newtheorem{bpa}[kcr]{Theorem}
\newtheorem{kpt}{Corollary}
\theoremstyle{plain}
\newtheorem{Remark: Possion}{Remark}
\newtheorem{Remark: Complexity}[Remark: Possion]{Remark}
\newtheorem{Remark: Decomposition}[Remark: Possion]{Remark}
\newtheorem{Remark: InfiniteK}[Remark: Possion]{Remark}
\newtheorem{Remark: HetDecompose}[Remark: Possion]{Remark}
\theoremstyle{plain}
\newtheorem{rpg}{Definition}

\begin{abstract}
Cloud radio access network {(C-RAN)} is proposed recently to reduce network cost, enable cooperative communications, and increase system flexibility through centralized baseband processing. By pooling multiple virtual base stations (VBSs) and consolidating their stochastic computational tasks, the overall computational resource can be reduced, achieving the so-called statistical multiplexing gain. In this paper, we evaluate the statistical multiplexing gain of VBS pools using a multi-dimensional Markov model, which captures the session-level dynamics and the constraints imposed by both radio and computational resources. Based on this model, we derive a recursive formula for the blocking probability and also a closed-form approximation for it in large pools. These formulas are then used to derive the session-level statistical multiplexing gain of both real-time and delay-tolerant traffic. Numerical results show that {VBS} pools can achieve more than $75\%$ of the maximum pooling gain with $50$ VBSs, but further convergence to the upper bound (large-pool limit) is slow because of the quickly diminishing marginal pooling gain, which is inversely proportional to a factor between the one-half and three-fourth power of the pool size. We also find that the pooling gain is more evident under light traffic load and stringent Quality of Service (QoS) requirement.
\end{abstract}

\begin{IEEEkeywords}
C-RAN, {VBS pooling}, statistical multiplexing.
\end{IEEEkeywords}

%
\IEEEpeerreviewmaketitle

\section{Introduction}
\IEEEPARstart{I}{n} recent years, the proliferation of mobile devices such as smart phones and tablets, together with the diverse applications enabled by mobile Internet, has triggered the exponential growth of mobile data traffic \cite{vni}. To accommodate the rapid traffic growth, cellular networks have been continuously evolving toward smaller cell size, wider bandwidth, and more advanced transmission technologies. However, the problems that arise, such as the increased interference and operational cost, are difficult to be solved via the traditional radio access network ({RAN}) architecture, in which the processing functionalities are packed into stand-alone base stations (BSs) and the cooperation between {BSs} is restricted by the limited inter-BS backhaul bandwidth.

To overcome the shortcomings of the traditional RAN architecture, cloud RAN {(C-RAN)} \cite{cran} is proposed with centralized baseband processing. {C-RAN} can facilitate the adoption of cooperative signal processing and potentially reduce the operational cost. A similar idea is also proposed under with the name wireless network cloud {(WNC)} \cite{wnc}. This kind of novel architecture has attracted substantial attentions recently: the key building blocks of {C-RAN} are investigated and its major use cases are identified \cite{ngmn, concert, 2015arXiv151207743S, nguyen2016sdn}. Centralized processing is combined with dynamical fronthaul switching to address the mobility and energy efficiency issues of small cells in \cite{colony,fluidnet}. Concerning realization-related issues, it is demonstrated in \cite{vbs,cloudiq, bigstation, pran, sdhcn} that BBU functionalities can be implemented as software, i.e. virtual base station (VBS), which runs on general purpose platforms (GPP). Compared with specialized-platform-based implementations, GPP-based implementation is more flexible in terms of the implementation of new functionalities and the management of computational resource. Further more, a {VBS} pool can be constructed by consolidating multiple VBSs to share the same set of computational resource. In this way, the computational resource can be utilized more efficiently and related cost can be reduced.

Despite the evident advantages of {C-RAN}, the massive bandwidth requirement of its fronthaul network poses a serious challenge: transmitting the baseband sample of a single $20$MHz {LTE} antenna-carrier ({AxC}) requires around 1Gbps link bandwidth \cite{cpri, fh}. Large-scale centralization will thus incur enormous fronthaul expenditure and potentially cancel out the gains. Fortunately, it is observed in \cite{cloudiq, multiplex} that substantial statistical multiplexing gain can be obtained even with small-scale centralization, justifying the deployment of small clusters of C-RAN. Yet, these results are obtained from simulations that are based on short-term small-scale traffic logs, and a generalized analytical model is in need to derive the optimal {VBS} cluster size. To this end, a session-level {VBS} pool model is proposed in \cite{gomez13} under the assumption of unconstrained radio resource and dynamic resource management. Here user sessions represent the time period in which users occupy computational resource in the pool. Nevertheless, this model does not reflect two realistic factors. Firstly, radio resource are often the main performance bottleneck for real networks, and thus the influence of radio resource should be reflected in the {VBS} pool model. Secondly, dynamic resource management, which re-assigns resources at the arrival and departure of each user session, may incur too much control overhead and overload the system \cite{cloudiq}. Hence, semi-dynamic resource management, which assigns resources on much larger time scales (e.g. hours to days) than the arrival and departure of user sessions (e.g. seconds to minutes), is more realistic. This assumption is also reasonable because the traffic statistics also vary in similarly large time scales, therefore the management plan designed for some traffic statistics can be useful for a fairly long period, and only need to be occasionally adjusted in the long run.

In our previous work \cite{vbsmodel}, we analyze the statistical multiplexing gain of homogeneous {VBS} pools, in which each VBS has the same traffic arrival, resource configuration, and service strategy. We derive a product-form expression for the stationary distribution of user sessions in each VBS and give a recursive method to compute the session blocking probability of the VBS pool. In this paper, we extend these results to heterogeneous {VBS} pools, in which there are multiple classes of VBSs with different session arrival, resource configurations, and service strategies. The computational complexity of the recursive method is also analyzed. Under the assumption of large {VBS} pools, we also derive a closed-form approximation for the blocking probability. We show through simulation that the approximation is precise even for medium-sized pools with around $50$ {VBSs}. We then use this approximation to quantitatively investigate the statistical multiplexing gain of VBS pools under the influence of different different factors, including pool sizes, VBS heterogeneity, traffic load, and the desired levels of {QoS}.

The main contributions of this paper are as follows:
\begin{itemize}
	\item
	We propose a realistic session-level model for heterogeneous {VBS} pools with both radio and computational resource constraints and semi-dynamic resource management. We show that this model constitutes a continuous-time multi-dimensional Markov chain and derive its product-form stationary distribution. We also illustrate how this model can be used to analyze real-time and delay-tolerant traffics. 
	\item
	We give a recursive method to compute the blocking probability for the proposed model. This method has quadratic computational complexity, much lower than brute-force evaluation which has exponential complexity. For large {VBS} pools, we also derive a closed-form formula to approximate the blocking probability.
	\item
	We provide an in-depth analysis on the statistical multiplexing gain of {VBS} pools. We show numerically the influence of various factors including the pool size, traffic load, and {QoS} requirements. We also prove that the statistical multiplexing gain increases slowly as the pool size grows large, with the residual gain diminishing at a speed between $|\bm{M}|^{-3/4}$ and $|\bm{M}|^{-1/2}$. Here $|\bm{M}|$ denotes the pool size.	
\end{itemize}

The rest of the paper is organized as follows. Section \ref{sec:model} introduces the proposed model and presents the product-form stationary distribution of user sessions. Section \ref{sec:pb} derives the recursive formula for the blocking probability and gives a closed-form approximation for large VBS pools. In Section \ref{sec:gain}, we derive the expression of statistical multiplexing gains and apply it to both real-time and delay-tolerant traffics in Section \ref{sec:eg}. Section \ref{sec:results} presents the numerical results and discusses the implications on realistic system design. Finally the paper is concluded in section \ref{sec:con}.

\section{Model Formulation} \label{sec:model}
In this section, we introduce the Markov model for {VBS} pools and derive its stationary distribution. The model captures the session-level dynamics in a {VBS} pool. To endow our model with enough generality, we assume $V$ different classes of {VBSs}. The total number of {VBSs} in class $v$  ($v = 1,2,\cdots,V$) is $M_v$. Each {VBS} in class $v$ is assigned with $K_v$ units of radio resource. To perform baseband signal processing on user sessions, all {VBSs} share a total of $N$ units of computational resource. The overall setting is illustrated in Fig. \ref{overview}.  We assume homogeneous resource demands: every active session is assumed to occupy one unit of radio resource and one unit of computational resource. Note that computational workloads that are independent of user dynamics do exist in cellular systems. However, the dominant consumers for computational resources are mostly per-user functions and thus the overall workload roughly follows a linear relationship with the number of users\cite{cloudiq}. For simplicity, hereafter we denote radio and computational resource by r-servers and c-servers, respectively.

    \begin{figure}[!t]
    	\centering
    	\includegraphics[width=3.5in, bb=56 374 527 634]{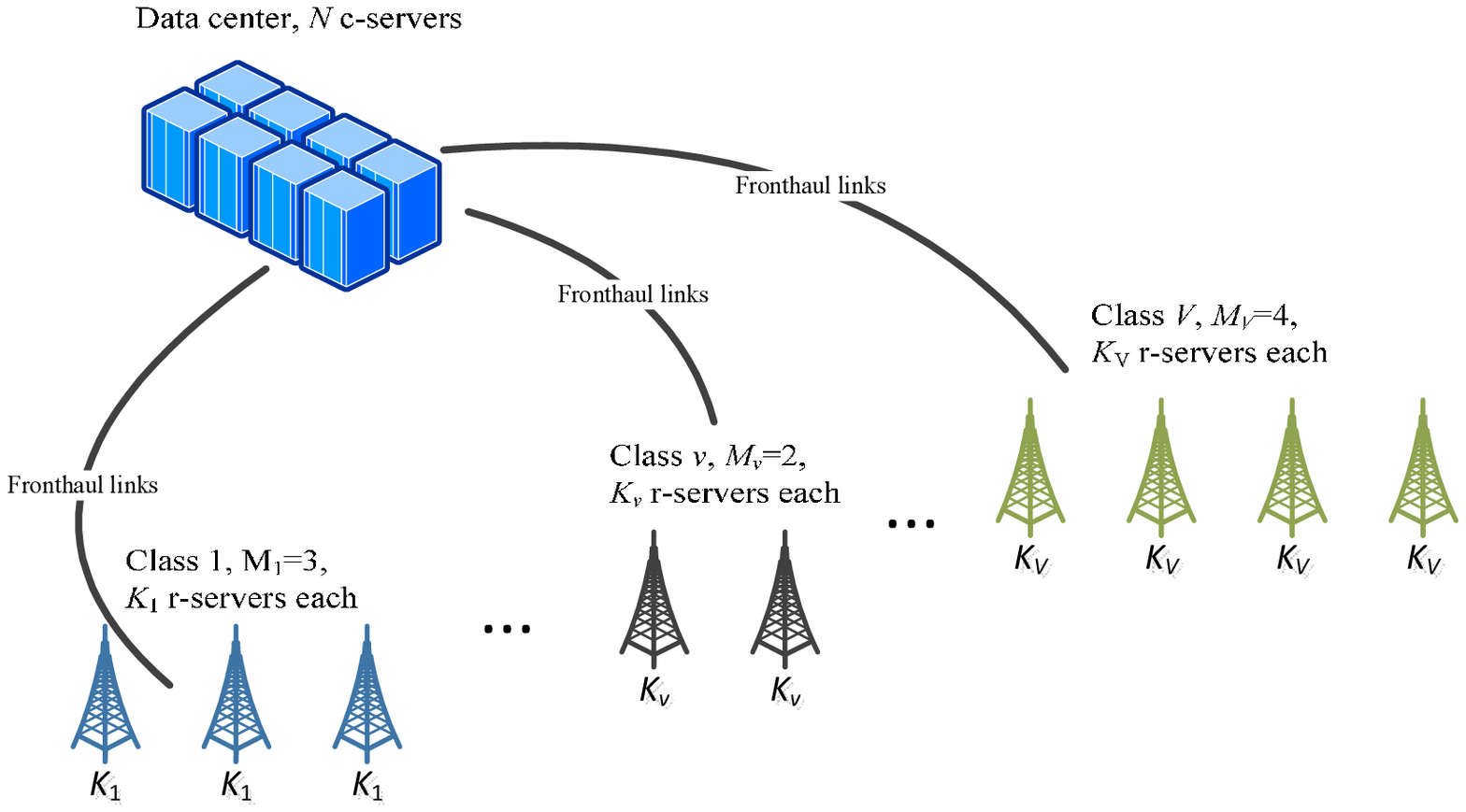}
    	\caption{An example of heterogeneous VBS pool. There are $V$ classes of VBSs, each class may have different number of VBSs and amount of radio resource. These VBSs are consolidated in a data center and shares $N$ units of computational resource.}
    	\label{overview}
    \end{figure}

\subsection{Session Arrival, Service Discipline, and Admission Control}\label{subsec:arrival}  
User sessions are initiated following independent Poisson processes in the coverage area of their serving {VBSs}. Obviously, the overall session arrival rate in a {VBS} is proportional to its coverage area. We allow {VBSs} in different classes to have different sizes of coverage areas, and consequently, different aggregated session arrival rates. We denote the arrival rate for class-$v$ {VBSs} by $\lambda_{v}$.

Each user session demands an exponentially distributed amount of service capacity before it leaves. Note the notion of service capacity can be flexibly interpreted according to the specific scenario this model is applied to. For example, service capacity can be considered as time duration for voice call sessions or the amount of information bits for cellular data sessions. For statistical QoS scenarios such as soft-real-time video, service capacity can still be interpreted as duration or information bits. The physical resources that enables such capacity is defined as the minimum amount of resource that can constantly satisfy a session. This means if the instantaneous requirement of a session is lower than provided, the remaining resources will become under-utilized. We assume that a {VBS} pool scheduler manages the service capacity so that the service capacity assigned to a class-$v$ {VBS} is a function of the total number of sessions in this {VBS}, and the assigned capacity is equally divided among these sessions for fairness  Denote the number of sessions in the $m$-th {VBS} of class-$v$ at time $t$ as $U_{v,m}(t)$. Then the above service strategy can be translated into a session departure rate function $f_v(U_{v,m}(t))$ for sessions in the $m$-th {VBS} of class-$v$ at time $t$. The above Poisson assumptions have been widely used in existing literature to evaluate the impact of randomness on system performance \cite{Borst05, bk_jsac}.  

To guarantee that active sessions always have enough r-servers and c-servers, the {VBS} pool has to enforce admission control on the arriving sessions: whenever a new session arrives, an admission control agent in the {VBS} pool will decide whether or not to accept this session according to the current resource usage. For class-$v$ sessions, the new session is accepted only if the number of r-servers in its serving {VBS} is less than $K_v$ and the number of c-servers in the pool is less than $N$; otherwise the session is blocked.

\subsection{State Transitions}   	
Recall that we denote the number of sessions in the $m$-th {VBS} of class-$v$ by $U_{v,m}(t)$, we can further describe the session dynamics in the {VBS} pool with a continuous-time stochastic process $$\bm{U}(t) = (U_{1,1}(t),\cdots,U_{1,M_1},\cdots,U_{V,1},\cdots,U_{V,M_V}(t))^T.$$ Given the Markovian property of the arrival and service of processes, it is obvious that $\bm{U}(t)$ is a Markov chain. Taking the admission control policy into consideration, we can get the set of possible system states
\begin{equation}\label{stateset}
\begin{aligned}
	\bm{U}(t) \in \mathbb{U} =  \{ & \bm{u} \mid 0 \le u_{v,m} \le K_v, \\
	& 0 \le \sum_{v=1}^{V}\sum_{m=1}^{M_v} u_{v,m} \le N, \\
	& u_{v,m} \in \mathbb{N}\},
\end{aligned}
\end{equation}
where $\bm{u} = (u_{1,1},\cdots,u_{1,M_1},\cdots,u_{V,1},\cdots,u_{V,M_V})^{T}$ is the state vector. Because the session arrivals and departures are Markovian, $\bm{U}(t)$ is a multi-dimensional birth-and-death process. The transition rate of $\bm{U}(t)$ from an arbitrary state $\bm{u^{(i)}}$ to another state $\bm{u^{(j)}}$ is:
\begin{equation}\label{transition}
	q_{\bm{u^{(i)}} \bm{u^{(j)}}} =
	\begin{cases}
        \lambda_v,         & \text{if $\bm{u^{(j)}}-\bm{u^{(i)}}= \bm{e_{v,m}}$}\\
        f_v(u_{v,m}^{(i)}), & \text{if $\bm{u^{(j)}}-\bm{u^{(i)}}= -\bm{e_{v,m}}$}\\
        0,               & \text{otherwise}
    \end{cases},
\end{equation}
where $u_{v,m}^{(i)}$ is the $(\sum\limits_{w=1}^{v-1}M_w + m)$-th entry of $\bm{u^{(i)}}$, and $$\bm{e_{v,m}} = (0,\cdots,0,\underbrace{1}_{(\sum\limits_{w=1}^{v-1}M_w + m)\text{-th}},0,\cdots,0)^T$$ is a column vector of length $ \sum\limits_{v=1}^{V}M_v$. For the ease of understanding, we illustrate the state transition graph of a simple example with $V=1$, $M_1=2$, $K_1=3$, $N=4$ and $f_1(n) = n\mu_0$ in Fig. \ref{2d}.
    
A similar problem has been formulated in the context of \emph{Stochastic Knapsack Problem}\cite{ross89}, which is a stochastic extension of the traditional knapsack problem. Specifically, the items which occupy a certain amount of space come into and leave a knapsack randomly. The model we formulate is mathematically equivalent to stochastic knapsacks under \emph{coordinate convex} \cite{AeinKosovych-37} admission control policies. However, the focus of these previous work was to find the policy that maximizes the reward of storing items, and the analysis was limited to problems with small dimensionality because the complexity increases dramatically as the number of item classes grows. In contrast, we aim at evaluating the blocking probability instead of reward, and we have to address the large-dimensionality problems due to the large sizes of {VBS} pools.
    \begin{figure}[!t]
    \centering
    \includegraphics[width=2.8in, bb=92 420 511 799]{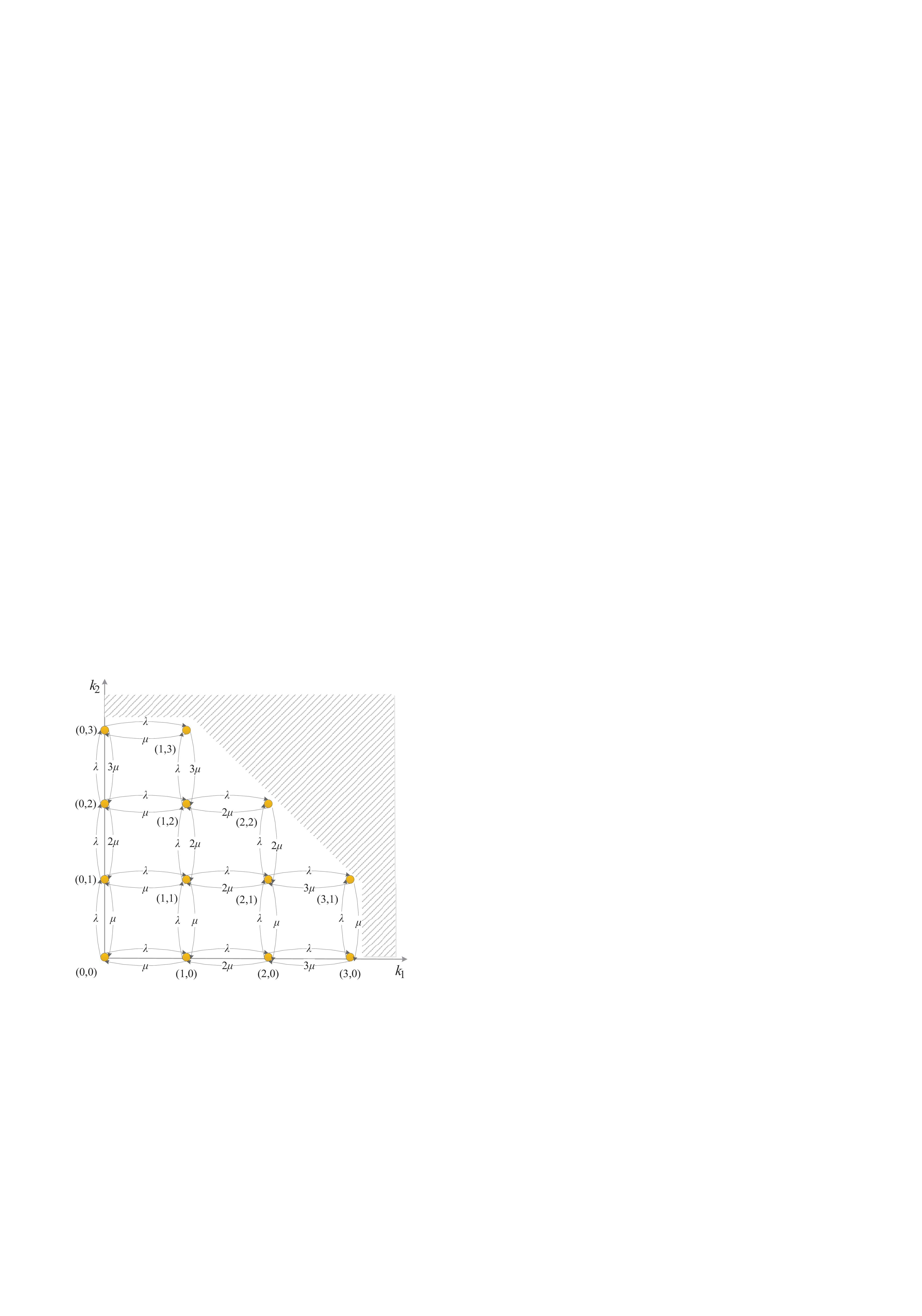}
    \caption{Transition graph of a {VBS} pool with two {VBSs}. The $k_1$ and $k_2$ axes indicate the number of active sessions in these two {VBSs}, respectively. Each yellow point represents a possible pool state, and the states in the gray region are prohibited because of the computational and radio resource constraints. ($K=3$, $N=4$, and $f_1(n) = n\mu_0$)}
    \label{2d}
    \end{figure}

\subsection{Stationary Distribution}
To perform further analysis, we need to derive the stationary distribution of $\bm{U}(t)$. Fortunately, the model we formulated guarantees the reversibility of $\bm{U}(t)$ as in the following theorem, which in turn results in a product-form expression for the stationary distribution.
\begin{rev}[Reversibility]\label{Theo_rev}
   A continuous-time Markov chain with a state set as in (\ref{stateset}) and transition rates as in (\ref{transition}) is reversible.
\end{rev}
The reversibility of $\bm{U}(t)$ has been proved in \cite{kaufman81} for more general cases. We provide an alternative proof in the Appendix using Kolmogorov's Criterion of Reversibility. Since $\bm{U}(t)$ is reversible, the local balance equation holds in the statistical equilibrium
        \begin{equation}\label{local}
		\begin{aligned}
		&\Pr\left\{ \bm{U}(\infty) = \bm{u^{(i)}} \right\}  q_{\bm{u^{(i)}}\bm{u^{(j)}}}\\
		& = \Pr\left\{ \bm{U}(\infty) = \bm{u^{(j)}} \right\}  q_{\bm{u^{(j)}}\bm{u^{(i)}}}.
		\end{aligned}
        \end{equation}
For simplicity, $\Pr\left\{ \bm{U}(\infty) = \bm{u} \right\}$ is hereafter denoted by $\Pr\left\{ \bm{u} \right\}$ or $\Pr\left\{ u_{1,1},\cdots,u_{1,M_1},\cdots,u_{V,1},\cdots,u_{V,M_V} \right\}$. Without loss of generality, let $\bm{u^{(i)}}$ and $\bm{u^{(j)}}$ be two arbitrary neighboring states:
        $$\bm{u^{(i)}} = (u_1,\cdots, u_{v,m},\cdots,u_{V,M_V})^T$$
        $$\bm{u^{(j)}} = (u_1,\cdots, u_{v,m}+1,\cdots,u_{V,M_V})^T,$$
and substituting (\ref{transition}) into (\ref{local}) yields:
        \begin{equation}\label{local2}
        \begin{aligned}
        & \Pr\left\{ u_1,\cdots, u_{v,m},\cdots,u_{V,M_V} \right\}  \lambda_v \\
        & = \Pr\left\{ u_1,\cdots, u_{v,m}+1,\cdots,u_{V,M_V} \right\}  f_v(u_{v,m} + 1).
        \end{aligned}
        \end{equation}
Clearly, this implies a recursive equation for computing the stationary distribution. Continuing the recursion down to $0$ for the $(\sum\limits_{w=1}^{v-1}M_w + m)$-th entry:
        \begin{equation}
        \begin{aligned}
        &\Pr\left\{ u_1,\cdots, u_{v,m}+1,\cdots,u_{V,M_V} \right\}\\
        &= \Pr\left\{ u_1,\cdots,0,\cdots,u_{V,M_V} \right\} \cdot \frac{\lambda_v^{u_m+1}}{\prod_{i=1}^{u_{v,m}+1} f_v(i)}.
        \end{aligned}
        \end{equation}
        Repeating the same process for other entries yields:
        \begin{equation}\label{static}
        \Pr\left\{ \bm{u} \right\} = P_0 \prod_{v=1}^{V}\prod_{m=1}^{M_v}\frac{\lambda_v^{u_{v,m}}}{\prod_{i=1}^{u_{v,m}} f_v(i)},
        \end{equation}
        in which
        \begin{equation}\label{static0}
\begin{aligned}
        P_0 & = \Pr\left\{ 0,\cdots,0,\cdots,0 \right\} \\
        & = \left( \sum_{\bm{u} \in \mathbb{U}}
        \prod_{v=1}^{V}\prod_{m=1}^{M_v}\frac{\lambda_v^{u_{v,m}}}{\prod_{i=1}^{u_{v,m}} f_v(i)}\right)^{-1}
\end{aligned}
        \end{equation}
is the probability of zero state and can be derived directly from the unity of probability distribution. As can be seen in (\ref{static}) and (\ref{static0}), the stationary distribution of any state $\bm{u}$ is proportional to the product of terms which can be solely determined by the entry values of $\bm{u}$.

\begin{Remark: Possion}
	Although we formulate our problem with the assumption of exponentially distributed service demand, it is worthy of noting that the above product-form stationary distribution also applies to other non-exponential service demand distributions. It is proved in \cite{kaufman81} that the product form distribution is valid for any service time distributions with rational Laplace transforms.
\end{Remark: Possion}

\section{Blocking Probability}\label{sec:pb}
\subsection{Brute-force Evaluation}
The admission control policy we have enforced on the {VBS} pool will cause session blockings. These blocking events can be classified into two classes: radio blockings (denoted by $B_r$) and computational blockings (denoted by $B_c$). Radio blocking is defined as the blocking events \emph{solely} due to insufficient r-servers, i.e. $U_{v,m}(t^{-}) = K, \text{and} \sum_{v=1}^{V}\sum_{m=1}^{M_v} U_{v,m}(t^{-}) < N,$
while computational blocking is defined as the blocking events due to insufficient c-servers \emph{regardless of r-servers}, i.e. $\sum_{v=1}^{V}\sum_{m=1}^{M_v} U_{v,m}(t^{-}) = N.$
Here $t^{-}$ means the epoch just prior to a session arrival.  Because we define radio blocking as the events that are solely due to insufficient r-servers, the blocking events that are due to simultaneously insufficient r-servers and c-servers are explicitly classified as computational blocking. It is worth noting that these doubly blocking events can be instead classified as radio blocking without affecting the overall blocking probability. But doing so will nevertheless result in less concise mathematical definition for both classes of events. Therefore, radio and computational blockings events are mutually exclusive, i.e. $B_r \cap B_c = \varnothing$. The union set of radio and computational blocking is further defined as the overall blocking $B = B_r \cup B_c$.

Next we derive the expression for the probability of radio and computational blockings. Since Poisson arrivals see time-averages ({PASTA}) \cite{PASTA}, the blocking probability can be evaluated from the stationary distribution we have just derived. Concretely, the radio blocking probability for sessions in class-$v$ {VBSs} is:
    
    \begin{align}
    P_v^{\mathrm{br}}
    &= \sum\limits_{m=1}^{M_{v}}\frac{1}{M_{v}} \sum\limits_{\bm{u} \in \mathbb{U}_{\mathrm{br}}^{v,m} }\Pr\left\{ \bm{u} \right\} \\
    &= P_0\sum\limits_{\bm{u} \in \mathbb{U}_{\mathrm{br}}^{v,1} }\prod_{w=1}^{V}\prod_{m=1}^{M_w}\frac{\lambda_w^{u_{w,m}}}{\prod_{i=1}^{u_{w,m}} f_w(i)} \label{pbr:4} \\
    &= P_0 \frac{\lambda_v^{K_v}}{\prod_{i=1}^{K_v} f_v(i)} \cdot \\
     &\sum\limits_{\bm{u} \in \mathbb{U}_{\mathrm{br}}^{v,1} }\left[(\prod_{w \neq v}\prod_{m=1}^{M_w}\frac{\lambda_w^{u_{w,m}}}{\prod_{i=1}^{u_{w,m}} f_w(i)})\cdot(\prod_{m=2}^{M_v}\frac{\lambda_v^{u_{v,m}}}{\prod_{i=1}^{u_{v,m}} f_v(i)})\right] \label{pbr},
    \end{align}
where $\mathbb{U}_{\mathrm{br}}^{v,m} = \{\bm{u} \mid u_{v,m} = K_v, u_{1,1}+\cdots+u_{1,M_1}+\cdots+u_{V,1}+\cdots+u_{V,M_V}<N\}$ and (\ref{pbr:4}) holds because (\ref{static}) is symmetric for entries with same values for index $v$, i.e.
	\begin{equation}
    \label{symm}
	\begin{aligned}
	    &\Pr\left\{ \cdots, u_{v,i},\cdots,u_{v,j},\cdots \right\}  \\ &= \Pr\left\{ \cdots, u_{v,j},\cdots,u_{v,i},\cdots \right\}.
	\end{aligned}
    \end{equation}
Similarly, the computational blocking probability is:
    \begin{equation}\label{pbc}
    \begin{aligned}
    P^{\mathrm{bc}}
    &= \sum_{\bm{u} \in \mathbb{U}_{\mathrm{bc}}^N}\Pr\left\{ \bm{u} \right\}\\
    &= P_0  \sum_{\bm{u} \in \mathbb{U}_{\mathrm{bc}}^N}\prod_{v=1}^{V}\prod_{m=1}^{M_v}\frac{\lambda_v^{u_{v,m}}}{\prod_{i=1}^{u_{v,m}} f_v(i)},
    \end{aligned}
    \end{equation}
where $\mathbb{U}_{\mathrm{bc}}^N = \{\bm{u} \mid u_{1,1}+\cdots+u_{1,M_1}+\cdots+u_{V,1}+\cdots+u_{V,M_V} =N, u_{v,m}\le K_v\}$. The overall blocking probability for class-$v$ VBSs can then be brute-force evaluated by summing up radio and computational blocking probability
    \begin{equation}\label{pbo}
        P_v^{\mathrm{b}} = \Pr\left\{ B \right\} = P_v^{\mathrm{br}}  + P^{\mathrm{bc}}.
    \end{equation}

\subsection{Recursive Evaluation} \label{recur}
Theoretically, the blocking probability under arbitrary system parameter can be calculated with brute-force evaluation. However the calculation process is exponentially hard and can become intractable when the pool size is extremely large. To reduce the computational complexity, we next give a recursive method for calculating the blocking probability. We will first introduce two auxiliary functions and re-express the blocking probability with respect to these functions. Then we will establish a recursive relationship to evaluate those two auxiliary functions and provide an analysis on the computational complexity of the proposed recursive evaluation method. These two auxiliary functions are defined as follows:
        \begin{equation}\label{cut}
        C(N,\bm{M}) = \sum\limits_{\bm{u} \in \mathbb{U}_{\mathrm{bc}}^N}\prod_{w=1}^{V}\prod_{m=1}^{M_w}\frac{\lambda_w^{u_{w,m}}}{\prod_{i=1}^{u_{w,m}} f_w(i)},
        \end{equation}
        \begin{equation}\label{remain}
        R(N,\bm{M}) = \sum\limits_{\bm{u} \in \left(\mathbb{U}_{\mathrm{bc}}^N\right)^C}\prod_{w=1}^{V}\prod_{m=1}^{M_w}\frac{\lambda_w^{u_{w,m}}}{\prod_{i=1}^{u_{w,m}} f_w(i)},
        \end{equation}
where $\bm{M} = (M_1,\cdots,M_v,\cdots,M_V)^T$ and $\left(\mathbb{U}_{\mathrm{bc}}^N\right)^C = \{\bm{u} \mid u_{1,1}+\cdots+u_{1,M_1}+\cdots+u_{V,1}+\cdots+u_{V,M_V}<N, u_{v,m} \le K_v\}$ is the complement set of set $\mathbb{U}_{\mathrm{bc}}^N$ in set  $\mathbb{U}$. Clearly, $C(N,\bm{M})$ and $R(N,\bm{M})$ are proportional to the sum of probability terms over $\mathbb{U}_{\mathrm{bc}}^N$ and $\left(\mathbb{U}_{\mathrm{bc}}^N\right)^C$, respectively. Therefore, the blocking probability (\ref{pbr}) and (\ref{pbc}) can be re-expressed as, 
        \begin{equation}
        \begin{aligned}
        &P_v^{\mathrm{br}}&= &P_0 \cdot \frac{\lambda_v^{K_v}}{\prod_{i=1}^{K_v} f_v(i)}   R(N-K_v,\bm{M}-\hat{\bm{e}}_{v}),\\
        &P^{\mathrm{bc}} &= &P_0 \cdot C(N,\bm{M}),\\
        &P_0 &= &R^{-1}(N + 1, \bm{M}),
        \end{aligned}
        \end{equation}       
where $\hat{\bm{e}}_v = (0,\cdots,0,\overbrace{1}^{v\text{-th}},0,\cdots,0)^T$ is a column vector of length $V$. From the definition of $C(N,\bm{M})$ and $ R(N,\bm{M})$, the following recursive relationships exist:
        \begin{equation}\label{rec_c}
        \begin{aligned}
        & C(N,\bm{M}) =\\
        & \begin{cases}
        \frac{\lambda_v^{N_2(v)}}{\prod_{i=1}^{N_2(v)} f_v(i)}, & \bm{M} = \hat{\bm{e}}_v\\
        \begin{aligned}
        \sum_{n=N_1(v)}^{N_2(v)}\frac{\lambda_v^{n}}{\prod_{i=1}^{n} f_v(i)}  C(N-n,\bm{M}-\hat{\bm{e}}_v), \end{aligned} & \bm{M} > \hat{\bm{e}}_v,
        \end{cases}
        \end{aligned}
        \end{equation}

        \begin{equation}\label{rec_r}
        \begin{aligned}
        & R(N,\bm{M}) \\
        & = \begin{cases}
        0, & N = 1 \\
        \begin{aligned} R(N+1,\bm{M})-C(N,\bm{M}),\end{aligned}   & 1 < N < \bm{M}^T\bm{K}+1\\
        \prod_{w=1}^{V} \left( \sum_{n=1}^{K_w}\frac{\lambda_w^{n}}{\prod_{i=1}^{n} f_w(i)} \right)^{M_w}, & N = \bm{M}^T\bm{K}+1,
        \end{cases}
        \end{aligned}
        \end{equation}
where $$N_1(v) = \max{\left[0,N-\sum_{w \neq v}M_wK_w-(M_v-1)K_v\right]},$$ $$N_2(v) = \min{(K_v,N)},$$ and $\bm{M}^T\bm{K} = \sum_{w=1}^{V}M_wK_w$. Following these recursive relationship, we can calculate the value $C(N,\bm{M})$ and $R(N,\bm{M})$ for arbitrary input through iterative calculation. Concretely, we can iterate for $C(N,\bm{M})$ from any  $\hat{\bm{e}_v}$ following (\ref{rec_c}). After that, we can reuse the calculated $C(N,\bm{M})$ values to calculate for $R(N,\bm{M})$ by iterating from either $N=0$ or $N=\bm{M}^T\bm{K}+1$ according to (\ref{rec_r}). Note that the comparison between $\bm{M}$ and $\hat{\bm{e}}$ in (\ref{rec_c}) is element-wise, and in this sense $\bm{M}$ is always greater or equal to $\hat{\bm{e}}$ in non-empty pools. With the above recursive relationship, the computational complexity of blocking probability can be reduced to at most the second power of the pool size, as stated in the following theorem.
\begin{complex}\label{Theo_complex}
The upper bound for the overall computational complexity of the proposed recursive method is:
\begin{equation}
C \le \left[(\max_v{K_v})^2+\max_v{K_v}\right]\cdot|\bm{M}|^2.
\end{equation}
\end{complex} 
\begin{proof}
See Appendix \ref{Proof: Complex} for proof.
\end{proof}
        
\subsection{Large Pool Approximation}
The above quadratic computational complexity can also become intractable for very large pools. To overcome this inconvenience, next we present a closed-form approximation for the blocking probability with large {VBS} pools. Although the above recursive expression cannot lead us to a direct approximation, the product-form stationary distribution of $\bm{U}$ does facilitate an indirect one.

First define some auxiliary variables. Let $\tilde{U}_{w,m}$ be the number of sessions in the $m$-th class-$w$ {VBS} when $$N\ge\bm{M}^T\bm{K},$$ $$\mu_w = \mathrm{E}\left[\tilde{U}_{w,m}\right],$$ $$\sigma_w^2 = \mathrm{Var}\left[\tilde{U}_{w,m}\right]\footnote{It is obvious that $\tilde{U}_{w,m}$ are i.i.d random variables for $m = 1,2,\cdots,M_w$.}.$$ Also, let $\tilde{S}_{\bm{M}} = \frac{1}{|\bm{M}|}\sum_{w=1}^{V}\sum_{m=1}^{M_w} \tilde{U}_{w,m}$, and $\tilde{S}_{M_w}=\frac{1}{M_w}\sum_{m=1}^{M_w}\tilde{U}_{w,m}$. Using these notations, the large-pool approximation is stated in the following Theorem.

\begin{bpa}[Large Pool Blocking Probability]\label{Theo_bpa}
For $N > |\bm{M}|\mu$, the session blocking probability for class-$v$ {VBSs} is:
\begin{equation}\label{bpa}
\lim\limits_{|\bm{M}| \to \infty}P_v^{\mathrm{b}} = \frac{1}{\sqrt{2 \pi |\bm{M}|\sigma^2}} \frac{1}{e^{\alpha^2/2}-1} + \tilde{P}_v^{\mathrm{br}},
\end{equation}
where  $\mu =\sum_{w=1}^{V}\beta_w\mu_w$, $\sigma^2 = \sum_{w=1}^{V}\sigma_w^2$, and $\beta_w = \lim\limits_{|\bm{M}| \to \infty}\frac{M_w}{|\bm{M}|}$; $\alpha = \frac{N - |\bm{M}|\mu}{\sqrt{|\bm{M}|}\sigma}$ is the normalized number of c-servers;  $\tilde{P}_v^{\mathrm{br}}$ is the overall blocking probability in class-$v$ {VBSs} when $N > \bm{M}^T\bm{K}$.
\end{bpa}
\begin{proof}
	See Appendix \ref{Proof: bpa} for proof.
\end{proof}

\begin{Remark: Complexity} With the approximation in (\ref{bpa}), the blocking probability can be calculated in one shot as long as the first-order and second-order statistics of $\tilde{U}_{v,m}$ are known. We will see in Section \ref{sec:eg} that these statistics are rather easy to obtain in some very practical scenarios.
\end{Remark: Complexity}

\begin{Remark: Decomposition}
Under the large-pool assumption, the blocking probability in (\ref{bpa}) is decomposed into two terms. The first term $\frac{1}{\sqrt{2 \pi |\bm{M}| \sigma^2}} \frac{1}{e^{\alpha^2/2}-1}$ reflects the portion of blockings that are solely due to insufficient computational resource, while the second term $\tilde{P}_v^{br}$ reflects the portion that are solely due to insufficient radio resource. This result reveals the decoupling feature between radio and computational blockings in large {VBS} pools.
\end{Remark: Decomposition} 

\begin{Remark: InfiniteK}
Although we assume $K_v<\infty$ in our derivation, (\ref{bpa}) is still true when $K \to \infty$. In this case, the approximation used in (\ref{Q_approx}) may not hold anymore. But this will not cause any problem since the radio blocking probability $\tilde{P}_v^{\mathrm{b}}$ will be $0$ when we have infinite radio resource. This will force the second term of (\ref{bpa}) to become zero, canceling out the inconsistency in the above approximation.
\end{Remark: InfiniteK}

\section{Statistical Multiplexing Gain}\label{sec:gain}
Since stochastic user sessions from different VBSs are consolidated, it is natural to expect a reduction in the required amount of computational resource compared with non-pooling schemes due to statistical multiplexing. Next we provide a theoretical analysis for the statistical multiplexing gain. We first derive the asymptotic utilization ratio of computational resource in large {VBS} pools.
\begin{lpl}[Large Pool Utilization Ratio]\label{Theo_lpl}
When c-servers are sufficiently provisioned (i.e. the number of c-servers is greater or equal to the number of r-servers, or $N \ge \bm{M}^T\bm{K}$), the utilization ratio of computational resource converges almost surely to a constant number that is smaller than 1 as $|\bm{M}| \to \infty$:
\begin{equation}
\lim\limits_{|\bm{M}| \to \infty}\eta \triangleq \frac{\sum_{w=1}^{V}\sum_{m=1}^{M_w} \tilde{U}_{w,m}}{N}
\xrightarrow{\mathrm{a.s.}}  \frac{|\bm{M}|\mu}{N} < 1.
\end{equation}
\end{lpl}
\begin{proof}
	See Appendix \ref{Proof: lpl} for proof.
\end{proof}
This theorem implies that there exist some ($1-\eta$) redundant computational resource when the {VBS} pool is large enough. Thus this limit can be seen as an the maximum portion of c-servers that one can turn down to save computational resource. The potential to further turn down c-servers can in turn be defined as the difference between current utilization ratio of c-servers and the large-pool limit $\eta$:

\begin{rpg}[Residual Pooling Gain]\label{rpg}
The residual pooling gain of a {VBS} pool is:
\begin{equation}
g_r \triangleq \frac{N}{\bm{M}^T\bm{K}} - \eta.
\end{equation}
\end{rpg}

Although some c-servers can be turned down due to the statistical multiplexing effect, the negative effect is that the overall blocking probability $P^{\mathrm{b}}$ will increase following (\ref{bpa}). Hence we have to trade {QoS} for the statistical multiplexing gain. This tradeoff will be favorable as long as the degradation in the {QoS} is not very significant. Using the results in Theorem \ref{Theo_bpa}, we can directly derive the following corollary to quantify such ``significance'' and approximate the gain of {VBS} pools of different sizes.

\begin{kpt}[Critical Tradeoff Point]\label{Theo_knee}
When $|\bm{M}| \to \infty$, the minimum computational resource $\alpha^*$ required to keep the overall blocking probability $P_v^{\mathrm{b}} \le \tilde{P}_v^{\mathrm{br}} + \delta$ ($\delta \approx 0$) for all $v$ is
\begin{equation}\label{knee}
\alpha^* = \sqrt{2\mathrm{ln}(\frac{1}{\sqrt{2 \pi |\bm{M}| \sigma^2 \delta^2}}+1)}.
\end{equation}
\end{kpt}
We will show later in Fig. \ref{allPbs}, that this critical tradeoff point is essentially the point where the blocking probability start to increase at a significantly higher speed. This type of points are often referred to ``knee points''.

The residual pooling gain at this critical tradeoff point is bounded as follows:
\begin{equation}\label{gr}
\begin{aligned}
g^*_r &= \frac{N - |\bm{M}|\mu}{\bm{M}^T\bm{K}} \\
&=  \sigma \frac{\alpha^*\sqrt{|\bm{M}|}}{\bm{M}^T\bm{K}} \in \frac{\alpha^*}{\sqrt{|\bm{M}|}}\cdot\left[\frac{\sigma}{\max_v{K_v}},\frac{\sigma}{\min_v{K_v}}\right],
\end{aligned}
\end{equation}
by which $g^*_r$ is roughly proportional to $\alpha^*/\sqrt{|\bm{M}|}$. Note $\alpha^*$ is also a function of the pool size $|\bm{M}|$, so $g^*_r$ is not necessarily proportional to $1/\sqrt{|\bm{M}|}$. Investigating two extreme cases will help to reveal the true relationship between $g^*_r$ and the pool size $|\bm{M}|$.

\emph{Extreme Case 1:} If $|\bm{M}|$ is not very large such that $\sqrt{2 \pi |\bm{M}| \sigma^2 \delta^2} \ll 1$ and $\sqrt{|\bm{M}|} \ll 1/\delta^2$, then $\alpha^*$ is approximately constant because:
        \begin{equation}\label{Msmall}
        \begin{aligned}
        \alpha^* 
        \approx & \sqrt{2\mathrm{ln}(\frac{1}{\sqrt{2 \pi |\bm{M}| \sigma^2 \delta^2}})} \\
        =&        \sqrt{\mathrm{ln}(\frac{1}{2 \pi \sigma^2}) + \mathrm{ln}(\frac{1}{\delta^2}) + \mathrm{ln}(\frac{1}{|\bm{M}|})}\\
        \approx & \sqrt{\mathrm{ln}(\frac{1}{2 \pi \sigma^2 \delta^2})}.
        \end{aligned}
        \end{equation}
In this case $g^*_r \propto |\bm{M}|^{-1/2}$, which decreases slowly with $|\bm{M}|$. Even so, considering the fact that the residual pooling gain is at most 1, we can still get considerable pooling gain with a small value of $|\bm{M}|$.

\emph{Extreme Case 2:}, if $|\bm{M}|$ is very large such that $\sqrt{2 \pi |\bm{M}| \sigma^2 \delta^2} \gg 1$, notice $\lim_{x \to 0} \mathrm{ln}(1+x) \approx x$:
\begin{equation}\label{Mlarge}
\begin{aligned}
\alpha^* \approx & \sqrt{2\frac{1}{\sqrt{2 \pi |\bm{M}| \sigma^2 \delta^2}}} \propto |\bm{M}|^{-1/4}.
\end{aligned}
\end{equation}
In this case $g^*_r \propto |\bm{M}|^{-3/4}$, which means that the decrease in the residual pooling gain will speed-up as $|\bm{M}|$ grows large.
\begin{Remark: HetDecompose}
As can be seen in (\ref{knee}), the critical point is invariant of the {VBS} class index $v$. This implies that the {VBS} heterogeneity is decomposed in large {VBS} pools. The reason for this phenomenon may be that in large {VBS} pools, the absolute number of c-servers is large. Therefore, different class of {VBSs} may tend to interfere less with each other.
\end{Remark: HetDecompose}

\section{Example Scenarios} \label{sec:eg}
In this section, we will apply the results derived so far to two specific scenarios: real-time and delay-tolerant traffic. For each scenario, we will first explain how they can be mapped to our model and then we will perform necessary derivations. Note in real-life systems, both type of traffic may exist at the same time. In such a case, the following analysis can still be applied if the available resources are divided to serve these two type of traffic separately.
\subsection{Real-time Traffic}
For real-time traffic such as voice calls, active sessions will constantly bring in signal processing workload. Therefore, dedicated r-servers and c-servers need to be provisioned upon admission to guarantee the {QoS} of active sessions. The service capacity in this scenario equals the temporal duration of sessions, which is not affected by the scheduling policy of {VBS} pools once the sessions are accepted. As a result, the departure rate function can be simplified as $ f_v(i) = i\mu_0$. The {QoS} target in this case is to keep the overall session blocking probability for class-$v$ {VBSs} under a certain small threshold $P_v^{\mathrm{bth}} \approx 0$. Obviously, the session dynamics in different {VBSs} are mutually independent when computational resource are sufficiently provisioned ($N>\bm{M}^T\bm{K}$). Therefore the radio blocking probability $\tilde{P}_v^{\mathrm{br}}$ can be calculated as 
\begin{equation}\label{pb}
\tilde{P}_v^{\mathrm{br}} = \frac{a_v^{K_v}}{K_v!}\left(\sum\limits_{i=0}^{K_v}\frac{a_v^i}{i!}\right)^{-1} \le  P_v^{\mathrm{bth}} \approx 0,
\end{equation}
where $a_v = \lambda_v/\mu_v$. Then the first-order and second-order statistics of $\tilde{U}_{v,m}$ can be approximated as follows:
    \begin{equation}\label{rt_mu}
    \begin{aligned}
    \mathrm{E}\left[\tilde{U}_{v,m}\right]
    =& \frac{\sum\limits_{i=0}^{K_v} i  \frac{a_v^i}{i!}}{\sum\limits_{i=0}^{K_v}\frac{a_v^i}{i!}}
    = \frac{a \sum\limits_{i=0}^{K_v-1} \frac{a_v^i}{i!}}{\sum\limits_{i=0}^{K_v}\frac{a_v^i}{i!}}\\
    =&a_v \left( 1- \frac{a_v^{K_v}}{K_v!}\left( \sum\limits_{i=0}^{K_v}\frac{a_v^i}{i!}\right)^{-1}\right)
    \approx a_v,
    \end{aligned}
    \end{equation}
    
    \begin{equation} 
    \begin{aligned}
    \mathrm{E}\left[\tilde{U}_{v,m}^2\right]
    &=  \frac{\sum\limits_{i=0}^{K_v} i^2  \frac{a_v^i}{i!}}{\sum\limits_{i=0}^{K_v}\frac{a_v^i}{i!}}    = \frac{a_v \sum\limits_{i=0}^{K_v-1} (i+1)  \frac{a_v^i}{i!}}{\sum\limits_{i=0}^{K_v}\frac{a_v^i}{i!}}\\
    &= \frac{a_v  (\sum\limits_{i=0}^{K_v-1}\frac{a_v^i}{i!} + a_v \sum\limits_{i=0}^{K_v-2} \frac{a_v^i}{i!})}{\sum\limits_{i=0}^{K_v}\frac{a_v^i}{i!}} \approx a_v + a_v^2.\\
    \end{aligned}
    \end{equation}     
Then 
\begin{equation}
\mu_v \approx a_v,
\end{equation}
\begin{equation}
\sigma_v^2 \approx a_v.
\end{equation}

\subsection{Delay-tolerant traffic}
For delay-tolerant traffic such as packet data, the pool scheduler can opportunistically divide the total service capacity among active sessions. Here we assume constant service capacity rate $ f_v(i) = \mu_v$ for class-$v$ {VBSs} and \emph{Proportional Fairness} scheduling algorithm which effectively divides the total service capacity equally among active sessions. Although this assumption manifests a processor sharing model, it is essentially equivalent to a Markovian queueing model with the same $\lambda_v$ and $\mu_v$. Note because sessions would require certain amount of radio resources for signaling, new sessions would be rejected if there're no more signaling radio channels regardless of the data channels left. Therefore even the sessions can wait they still cannot be admitted into the system. If the rejected session decides to wait and retry, it will be considered as a new session. What’s more, many delay-tolerant traffic or elastic traffic still have a minimum rate requirement, which also limits the number of sessions that can be simultaneously served by the system. 
    
To derive the statistics in this scenario, first let $a_v = \lambda_v / \mu_v$ be the traffic load of the {VBSs} and define the following auxiliary function $A(a,K)$:
        \begin{equation}\label{A}
        \begin{aligned}
        A(a,K) &= \sum\limits_{i=0}^{K} a^i = \frac{1-a^{K+1}}{1-a},\\
        A'_a(a,K) &= \left(\sum\limits_{i=0}^{K} a^i\right)_a' = \sum\limits_{i=1}^{K} i a^{i-1} \\
        &= \frac{1-(K+1) a^K+K a^{K+1}}{(1-a)^2},\\
        A''_a(a,K) &= \left(\sum\limits_{i=0}^{K} a^i\right)_a'' = \sum\limits_{i=2}^{K} i(i-1) a^{i-2}.
        \end{aligned}
        \end{equation}
With these definitions, the average and covariance of $\tilde{U}_{v,m}$ can be expressed as:
        \begin{equation}\label{ps_mu}
        \begin{aligned}
        \mathrm{E}\left[\tilde{U}_{v,m}\right]
        &= \frac{\sum\limits_{i=1}^{K_v} i a_v^i}{\sum\limits_{i=0}^{K_v} a_v^i} = \frac{a_v A_a'(a_v,K_v)}{A(a_v,K_v)},
        \end{aligned}
        \end{equation}

        \begin{equation}\label{ps_sigma}
        \begin{aligned}
        \mathrm{E}\left[\tilde{U}_{v,m}^2\right]
        &= \frac{\sum\limits_{i=1}^{K_v} i^2 a_v^i}{\sum\limits_{i=0}^{K_v} a_v^i} = \frac{\sum\limits_{i=1}^{K_v} i a_v^i + \sum\limits_{i=2}^{K_v} i(i-1) a_v^i}{\sum\limits_{i=0}^{K_v} a_v^i}\\
        &= \frac{a_vA_a'(a_v,K_v) + a_v^2 A_a''(a_v,K_v) }{A(a_v,K_v)}.
        \end{aligned}
        \end{equation}
Again when computational resource are sufficiently provisioned ($N>\bm{M}^T\bm{K}$), we have
    \begin{equation}\label{ps blk}
        \tilde{P}_v^{\mathrm{br}} = \frac{a_v^k}{\sum\limits_{i=0}^{K_v} a_v^i} = \frac{a_v^k}{A(a_v,K_v)}.
    \end{equation}
Although these formulas are already enough for us to evaluate the performance of a {VBS} pool, the evaluation is nevertheless quite cumbersome. To simplify these formulas, we further assume that $K_v$ for all $v$ are large enough such that $K_v^2 a_v^{K_v} \to 0$.\footnote{This assumption is realistic because $a_v^{K_v}$ will diminish exponentially when $a_v<1$. Thus $K_v^2 a_v^{K_v}$ will be driven to $0$ for large enough $K_v$.} In this regime, 
	    \begin{equation}\label{Aapprox}
	    \begin{aligned}
	    A(a,K) &\approx \frac{1}{1-a}, \\
	    A_a'(a,K) &\approx \frac{1}{(1-a)^2},\\
	    A_a''(a,K) &\approx \frac{2}{(1-a)^3}.    
	    \end{aligned}
	    \end{equation}
Using (\ref{Aapprox}), (\ref{ps_mu}) and (\ref{ps_sigma}) can be simplified as
		\begin{equation}\label{ps_mu_approx}
		\begin{aligned}
		\mathrm{E}\left[\tilde{U}_{v,m}\right] \approx \frac{a_v}{1-a_v},
		\end{aligned}
		\end{equation}
		
		\begin{equation}\label{ps_sigma_approx}
		\begin{aligned}
		\mathrm{E}\left[\tilde{U}_{v,m}^2\right] \approx \frac{a_v}{1-a_v} + \frac{2a_v^2}{(1-a_v)^2}.
		\end{aligned}
		\end{equation}
Thus 
\begin{equation}
\mu_v \approx \frac{a_v}{1-a_v},
\end{equation}
\begin{equation}
\sigma_v^2 \approx \frac{a_v}{1-a_v} + \frac{a_v^2}{(1-a_v)^2}.
\end{equation}

\section{Numerical Results} \label{sec:results}
In this section, we will use the recursive method to numerically evaluate the blocking probability and compare them with the large-pool approximations.
\subsection{Basic Characteristics}
Fig. \ref{allPbs} shows the exact and large-pool-approximated blocking probability of a {VBS} pool under real-time traffic and different number of c-servers, and Fig. \ref{allPbs_delay} shows the same metrics for a {VBS} pool under delay-tolerant traffic. Note x-axis is normalized by the number of c-servers required without pooling to show the relative pooling gain. As can be seen, the trend in both figures are similar. This coincide with our large-pool approximation results that the blocking probability are affected only by the first- and second-order statistics of the number of sessions in the {VBS} pool. For this reason, we will only present results for real-time traffic from now on, and the conclusions we draw should apply to the delay-tolerant case as well. 

We can observe some basic blocking characteristic from these two figures: 1) when the number of c-servers is sufficient, the computational blocking probability $P^{\mathrm{bc}}$ is very small and below the scope of this figure; the overall blocking probability is dominated by radio blocking probability $P^{\mathrm{br}}$, which is around the desired threshold $P^{\mathrm{bth}}$. 2) As the number of c-servers decreases from its largest value $\bm{M}^T\bm{K}$, the computational blocking probability increases rapidly while the radio blocking probability start to decrease slightly; the net result of these two trends is a plateau before the critical tradeoff point (``knee point'') and a significant increase after it. 3) If the number of c-servers is to further decrease, the overall blocking probability will be dominated by the computational blocking probability and saturates at probability $1$. The radio blocking probability will decrease rapidly and its influence on overall blocking probability will diminish. Fig. \ref{allPbs} and Fig. \ref{allPbs_delay} also show the approximated blocking probability. The fact that the ``knee point'' configuration saved more than $20\%$ computational resources with a penalty of only $10^{-4}$ increase in the blocking probability demonstrates the benefit of statistical multiplexing. the As expected, these approximations are coherent with the exact values.

    \begin{figure}[!t]
    \centering
    \includegraphics[width=3.3in]{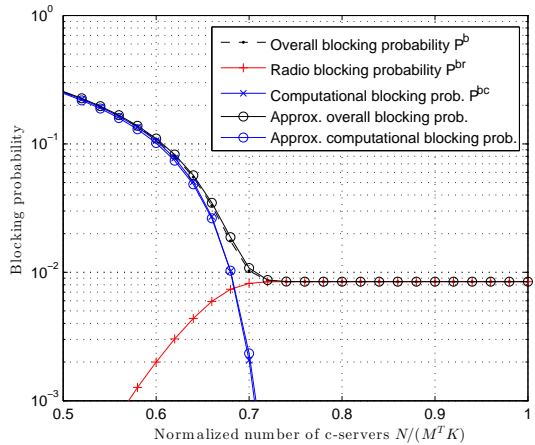}
    \caption{Blocking probability of a homogeneous VBS pool as a function of normalized number of c-servers ($N/\bm{M}^{T}\bm{K}$) under real-time traffic. Black box indicates the knee point. Simulation parameters: $M_1=40$, $a_1=20$, $P_1^{\mathrm{bth}}=10^{-2}$, $K_1=30$.}
    \label{allPbs}
    \end{figure}

    \begin{figure}[!t]
    \centering
    \includegraphics[width=3.5in]{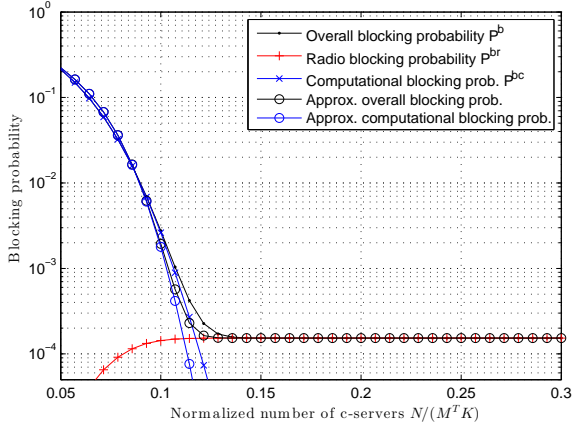}
    \caption{Blocking probability of a homogeneous VBS pool as a function of normalized number of c-servers ($N/\bm{M}^{T}\bm{K}$) under delay-tolerant traffic. Simulation parameters: $M_1=100$, $a_1=0.5$, $P_1^{\mathrm{bth}}=5\times10^{-4}$, $K_1=10$.}
    \label{allPbs_delay}
    \end{figure}

\subsection{Heterogeneous VBSs}
In Fig. \ref{ExactPbM2020}, we show the blocking probability of a {VBS} pool with two class of {VBSs}. The two classes have the same number of {VBSs} and the same traffic load, but the {QoS}, and thus the number of provisioned r-servers, is different. From the curves we can observe similar basic blocking characteristics as in the single class case. Also, we can see that the overall blocking probability for the two classes are different: they are respectively close to their threshold blocking probability when c-servers are sufficient since the overall blocking probability are dominated by the radio blocking, and converges to the same curve when c-servers become insufficient because the computational blocking probability begins to overwhelm. 
    \begin{figure}[!t]
    \centering
    \includegraphics[width=3.5in]{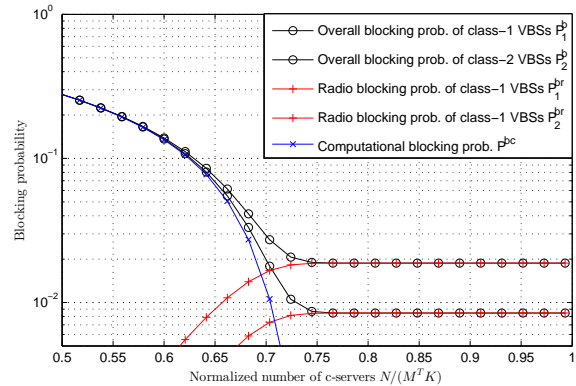}
    \caption{Blocking probability of a heterogeneous VBS pool as a function of normalized number of c-servers ($N/\bm{M}^{T}\bm{K}$) under real-time traffic. Simulation parameters: $\bm{M}=[20,20]^{T}$, $\bm{a}=[20,20]^{T}$, $\bm{P}^{\mathrm{bth}}=[1,2]^{T}\times10^{-2}$, $\bm{K}=[30,28]^{T}$.}
    \label{ExactPbM2020}
    \end{figure}

\subsection{Influence of Traffic Load and QoS Target}
In Fig. \ref{PbMultiak}, we illustrate the influence of different traffic load and {QoS} target. As can be seen, the desired level of {QoS} ($P^{\mathrm{bth}}$) determines the minimum blocking probability(i.e. height of the ``plateau'' to the right of the figure); while the traffic load $a$ determines the position of the ``knee point'' and how fast the blocking probability saturates to $1$. 
    \begin{figure}[!t]
    \centering
    \includegraphics[width=3.5in]{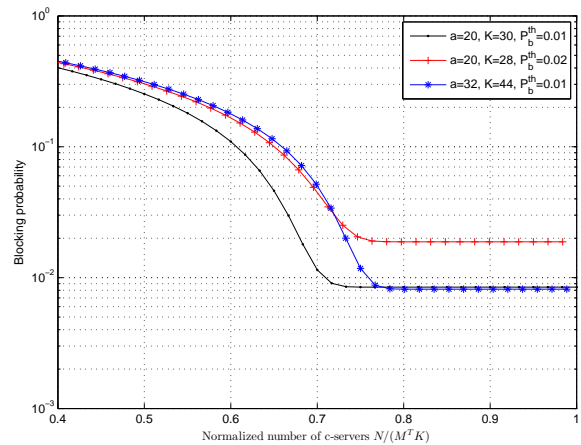}
    \caption{Blocking probability of a homogeneous VBS pool as a function of normalized number of c-servers ($N/\bm{M}^{T}\bm{K}$) under real-time traffic with different traffic load and {QoS} guarantees. Pool size $M=40$.}
    \label{PbMultiak}
    \end{figure}    

\subsection{Statistical Multiplexing Gain}
Most importantly, we can quantify the statistical multiplexing gain of the simulated {VBS} pool with the equations derived previously. In Fig. \ref{GainThreeFig}, we compare the overall blocking probability of three {VBS} pools under varying pool size. The ``knee point'' position and large-pool limit are also marked out with vertical lines. As the pool size increases, the blocking probability curve (and so does the ``knee point'') is pushed to the left. But the closer the ``knee point'' is to the large-pool limit, the slower the remaining distance decreases with the pool size $|\bm{M}|$. This indicates a decreasing marginal statistical multiplexing gain. Comparing the curves, we can find that the traffic load and the desired level of {QoS} have influence on the blocking probability and the statistical multiplexing gain. 

To better investigate this influence, we show the knee point position versus varying pool size in Fig. \ref{GainThreeFig_2}. Firstly we can find that a medium sized {VBS} pool can readily obtain considerable statistical multiplexing gain and the marginal gain diminishes fast. Thus a huge number of {VBSs} is needed so that the pooling gain can approach the large-pool limit. These observations imply that a {C-RAN} formed with multiple medium sized {VBS} pools can obtain almost the same pooling gain as the one formed with a single huge pool. If we further take the expenditure of fronthaul network into consideration, the former choice may be far more economical than the latter one.

By contrasting the left and middle curves, we can see that stricter {QoS} requirements can increase the pooling gain. This is because on one hand, we need to increase the number of r-servers $\bm{K}$ in order to reduce the blocking probability, which will in turn  increase $\bm{M}^T\bm{K}$; on the other hand, the average number of c-servers occupied is always around $|\bm{M}|\mu$. Therefore the stricter the {QoS}, the more idle c-servers there will be in the {VBS} pool and consequently the more the pooling gain. Also, we can see that the increase in traffic load will reduce the pooling gain by pushing the ``knee point'' to larger values. This observation indicates that, we may need to dynamically adjust the size\footnote{This can be achieved by dynamically changing the switching configuration of fronthaul so that the traffic of {VBSs} can be sent to a data center of the desired size.} of {VBS} pools in order to get a satisfactory pooling gain under fluctuating traffic load.

    \begin{figure*}[!t]
    \centering
    \includegraphics[width=6in]{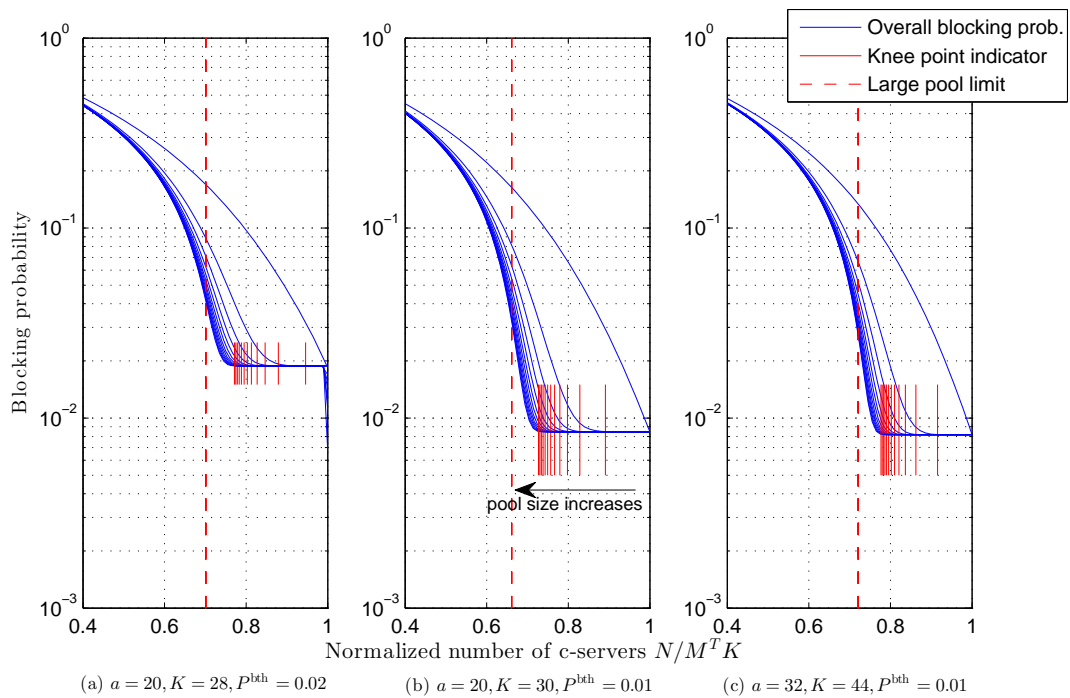}
    \caption{Blocking probability in homogeneous VBS pools as a function of normalized number of c-servers under real-time traffic and different pool size. Red vertical line indicates the knee point position. Red dashed line indicates the large-pool limit. Curves to the left correspond to larger pool size.}
    \label{GainThreeFig}
    \end{figure*} 
    
    \begin{figure*}[!t]
    \centering
    \includegraphics[width=6in]{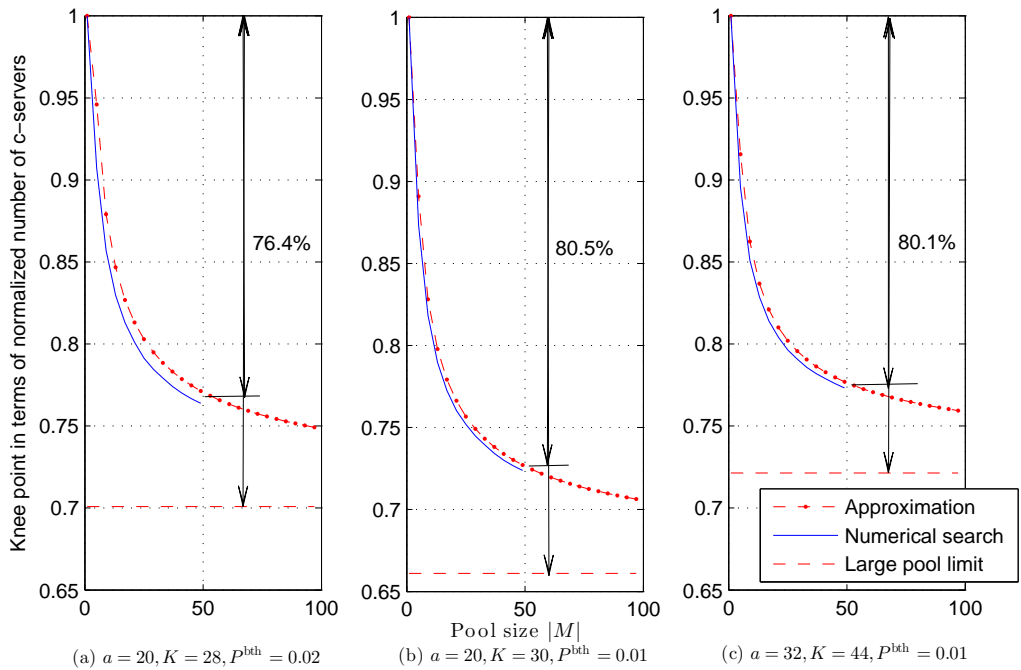}
    \caption{Knee point position as a function of pool size. Knee point position is measured in terms of the number of c-servers required normalized by the number of c-servers needed without pooling. Red dotted line shows the knee point values derived from large-pool approximation, whereas blue line shows the values directly searched from numerical results. Red dashed lines represent the large-pool limit $\eta$ as pool size approaches infinity. The percentage of the pooling gain at $50$ VBSs with respect to the maximum possible gain ($1-\eta$) is also noted in the figure.}
    \label{GainThreeFig_2}
    \end{figure*}

\section{Conclusion}
\label{sec:con}
In this article, we proposed a multi-dimensional Markov model for {VBS} pools to analyze their statistical multiplexing gain. We showed that the proposed model have a product-form expression for the stationary distribution. We derived a recursive method for calculating the blocking probability of a {VBS} pool, and gave closed-form approximation when the pool is large enough. Based on these results, we derived the expressions for the statistical multiplexing gains and applied them to both real-time and delay-tolerant traffic. Numerical results reveal that 1) the pooling gain reaches a significant level even with medium pool size (more than $75\%$ of the pooling gain can be achieved with around $50$ VBSs); 2) the marginal gain of larger pool size tend to be negligible; 3) lighter traffic load and tighter {QoS} level can increase the pooling gain.

Our model can be extended in several aspects to accommodate for more general scenarios. Firstly, we assume that user sessions occupy equal and fixed amounts of radio and computational resources. Yet realistic resource scheduling algorithm may allocate different amount of resources for each individual session. To accommodate such cases, our model need to be further relaxed to allow state transitions among non-neighboring states. Secondly, we assume sessions are only attached to one cell of the system. Nevertheless, coordinated-multipoint (CoMP) transmission/reception may introduce sessions that simultaneously consume radio resources from multiple cells. This means the admission control of CoMP session is hinged upon the available radio resources in all its serving cells, which can be accounted for by introducing more comprehensive admission controls in the model. Thirdly, we assume session arrival and service to be Poission. However, many emerging types of multi-media traffic is found to exhibit certain burstiness. The influence of burstiness can be investigated by assuming more general stochastic traffic and service models, e.g Interupted Poisson Process \cite{wu13}. Fourthly, we investigated real-time and delay-tolerant traffic separately whereas they are likely to coexist in real system. To evaluate such heterogeneous traffic, our model need to be extended to account for the different resource usage patterns of the two session types. Last but not the least, our resource reservation model may not be the most efficient possible. For example, unused service capacity can be further shared to increase statistical multiplexing gain. Regarding this, our model need to be further refined to support more general admission control and service strategies.

\appendices
\section{Reversibility of Proposed Model} \label{sec:rev} 
To proof that $\bm{U}(t)$ is reversible, we first give the Kolmogorov's Criterion of Reversibility.
	\begin{kcr}[Kolmogorov's Criterion]
	A continuous-time Markov chain is reversible if and only if its transition rates satisfy
	    \begin{eqnarray}
	    & q_{\bm{u^{(1)}} \bm{u^{(2)}}} q_{\bm{u^{(2)}} \bm{u^{(3)}}} \cdots q_{\bm{u^{(n-1)}} \bm{u^{(n)}}} q_{\bm{u^{(n)}} \bm{u^{(1)}}} \label{kcrt1} \\
	    =& q_{\bm{u^{(1)}} \bm{u^{(n)}}} q_{\bm{u^{(n)}} \bm{u^{(n-1)}}}
	    \cdots q_{\bm{u^{(3)}} \bm{u^{(2)}}} q_{\bm{u^{(2)}} \bm{u^{(1)}}} \label{kcrt2}
	    \end{eqnarray}
	for all finite sequences of states $\bm{u^{(1)}},\bm{u^{(2)}},\cdots,\bm{u^{(n)}} \in \mathbb{U}$.
	\end{kcr}
We next verify that the Kolmogorov's Criterion is satisfied by $\bm{U}(t)$ for any finite sequences of states $\bm{u^{(1)}},\bm{u^{(2)}},\cdots,\bm{u^{(n)}} \in \mathbb{U}$. 

\emph{Case 1}. If (\ref{kcrt1}) (or (\ref{kcrt2})) equals to $0$, then there must be at least one product term being $0$ in (\ref{kcrt1}) (or (\ref{kcrt2})). Assuming this term to be $q_{\bm{u^{(i)}} \bm{u^{(i+1)}}}$, then according to (\ref{transition}), the term $q_{\bm{u^{(i+1)}} \bm{u^{(i)}}}$ in (\ref{kcrt2}) (or (\ref{kcrt1})) must also be $0$. Thus (\ref{kcrt2}) (or (\ref{kcrt1})) must also equal to $0$.

\emph{Case 2}.  If neither (\ref{kcrt1}) nor (\ref{kcrt2}) is $0$, then none of the terms in (\ref{kcrt1}) and (\ref{kcrt2}) equals $0$. According to (\ref{transition}), the transition resulting in the term $q_{\bm{u^{(i)}} \bm{u^{(i+1)}}}$ must be between neighboring states, i.e.
    $\bm{u^{(i+1)}} - \bm{u^{(i)}} = (0,\cdots,0,\pm 1,0,\cdots,0)^T.$

Notice that a transition such that the $m$-th entry of a state is changed from $u_{m}$ to $u_{m}+1$ will produce a product term $\lambda$ in (\ref{kcrt1}) and a product term $ f(u_m+1)$ in (\ref{kcrt2}). And reversely, a transition such that the $m$-th entry is changed from $u_{m}+1$ to $u_{m}$ will produce a product term $ f(u_m+1)$ in (\ref{kcrt1}) and a product term $\lambda$ in (\ref{kcrt2}).

If we denote the number of state transitions in the transition loop
\begin{equation} \label{loop}
    (\bm{u^{(1)}},\bm{u^{(2)}},\cdots,\bm{u^{(n)}},\bm{u^{(1)}})
\end{equation}
such that the $(\sum\limits_{w=1}^{v-1}M_w + m)$-th state entry is changed from $u_{v,m}$ to $u_{v,m}+1$ by $n_{v,m,u_{v,m}}^{+}$ times, and the number of transitions such that the $(\sum\limits_{w=1}^{v-1}M_w + m)$-th entry is changed from $u_{v,m}+1$ to $u_{v,m}$ by $n_{v,m,u_{v,m}}^{-}$ times, then

\begin{equation}
\label{leqn}
\begin{aligned}
&\text{Eq.} (\ref{kcrt1}) \\
&= \prod_{w=1}^{V}\prod_{m=1}^{M_w}\prod_{u_{w,m}=1}^{K} \lambda^{n_{w,m,u_{w,m}}^{+}}  f_w^{n_{w,m,u_{w,m}}^{-}}(u_{w,m}+1),
\end{aligned}
\end{equation}

\begin{equation}
\label{reqn}
\begin{aligned}
& \text{Eq.} (\ref{kcrt2}) \\
&= \prod_{w=1}^{V}\prod_{m=1}^{M_w}\prod_{u_{w,m}=1}^{K} \lambda^{n_{w,m,u_{w,m}}^{-}}  f_w^{n_{w,m,u_{w,m}}^{+}}(u_{w,m}+1).
\end{aligned}
\end{equation}

To make (\ref{loop}) a closed loop, we must have 
\begin{equation}
\label{neqn}
n_{v,m,u_{v,m}}^{+} = n_{v,m,u_{v,m}}^{-}
\end{equation}
for all $v = 1,2,\cdots,V$, $m = 1,2,\cdots,M_v$ and $u_{v,m} = 1,2,\cdots,K$. Substituting (\ref{neqn}) into (\ref{leqn}) and (\ref{reqn}) and we get that (\ref{kcrt1}) equals to (\ref{kcrt2}). Thus $\bm{U}(t)$ is reversible.

\section{Proof of Theorem \ref{Theo_complex}}\label{Proof: Complex}
The calculation of a single $C(N,\bm{M})$ term involves at most $K_v$ summations. Also, there are at most $N \le \bm{M}^T\bm{K}$ such terms for some specific $\bm{M}$. Hence the computational complexity for all $C(N,\bm{M})$ terms for a {VBS} pool sized $\bm{M}$ is bounded as
\begin{equation}
\begin{aligned}
C_1 &\le (\max_v{K_v})\bm{M}^{T}\bm{K}|\bm{M}|\le (\max_v{K_v})^2|\bm{M}|^2,
\end{aligned}
\end{equation}
where $|\bm{M}| = \sum_{w=1}^VM_w.$ At the same time, the calculation of $R(N,\bm{M})$ involves only a single subtraction (or summation). Again, there are at most $\bm{M}^T\bm{K}$ such terms for some specific $\bm{M}$. Therefore the computational complexity of $R(N,\bm{M})$ for a {VBS} pool with size vector $\bm{M}$ is bounded as
\begin{equation}
\begin{aligned}
C_2 &\le \bm{M}^T\bm{K}|\bm{M}|\\
&\le (\max_v{K_v})|\bm{M}|^2.
\end{aligned}
\end{equation}
All together, the overall computational complexity is bounded as
\begin{equation}
\begin{aligned}
C &= C_1 + C_2\\
&\le \left[(\max_v{K_v})^2+\max_v{K_v}\right]|\bm{M}|^2.
\end{aligned}
\end{equation}
This bound is essentially quadratic in the pool size $|\bm{M}|$. Therefore, the computational complexity of blocking probability should also be quadratic in the pool size.

\section{Proof of Theorem \ref{Theo_bpa}} \label{Proof: bpa}
From the definition of $\tilde{S}_{\bm{M}}$ we know
\begin{equation}
\begin{aligned}
\lim\limits_{|\bm{M}| \to \infty}\tilde{S}_{\bm{M}} &= \lim\limits_{|\bm{M}| \to \infty}\frac{1}{|\bm{M}|}\sum_{w=1}^{V}\sum_{m=1}^{M_w} \tilde{U}_{w,m}\\
&=\lim\limits_{|\bm{M}| \to \infty}\sum_{w=1}^{V}\frac{M_w}{|\bm{M}|}\frac{\sum_{m=1}^{M_w} \tilde{U}_{w,m}}{M_w}\\
&=\lim\limits_{|\bm{M}| \to \infty}\sum_{w=1}^{V}\beta_w \tilde{S}_{M_w}.\\
\end{aligned}
\end{equation}
According to the Central Limit Theorem, $\tilde{S}_{M_w}$ converges in distribution to a normal random variables as $|\bm{M}| \to \infty$:
\begin{equation}
\lim\limits_{M_v \to \infty}\tilde{S}_{M_w} \sim N(\mu_w,\frac{\sigma_w^2}{M_w}).
\end{equation}
Since $\tilde{U}_{v,m}$ are independent random variables for all $v$ and $m$,  $\tilde{S}_{M_v}$ are also independent. Therefore $\tilde{S}_{\bm{M}}$ will also converge to a normal distributed random variable:
\begin{equation}\label{sm}
\begin{aligned}
\lim\limits_{|\bm{M}| \to \infty}\tilde{S}_{\bm{M}} 
&= \lim\limits_{|\bm{M}| \to \infty}\sum_{w=1}^{V}\beta_w \tilde{S}_w\\
&\sim N(\sum_{w=1}^{V}\beta_w \mu_w,\sum_{w=1}^{V}\beta_w\frac{\sigma_w^2}{M_w})\\
&\sim N(\mu,\frac{\sigma^2}{|\bm{M}|}).
\end{aligned}
\end{equation}

To express the blocking probability in terms of this normal distribution, we next establish a relationship between the stationary distributions of $\bm{U}$ and $\bm{\tilde{U}}$. Let $\tilde{P}_0$ be the probability of zero state for $\bm{\tilde{U}}$, then from the product-form stationary distribution of $\bm{U}$ and $\bm{\tilde{U}}$ we can get the following scaling relationship between the stationary distribution of $\bm{U}$ and $\bm{\tilde{U}}$:
\begin{equation} \label{scaling}
	\frac{\Pr\left\{ \bm{U} = \bm{u} \right\}}{P_0} =  \frac{\Pr\left\{ \bm{\tilde{U}} = \bm{u} \right\}}{\tilde{P}_0}.
\end{equation}
From the definition of $P_0$ and $\tilde{P}_0$, the following relationship exists:
\begin{equation}
\begin{aligned}
\frac{P_0}{\tilde{P}_0} &= \Pr\left\{ \sum_{w=1}^{V}\sum_{m=1}^{M_w} \tilde{U}_{w,m} \le N \right\}^{-1}\\
&= \Pr\left\{ \frac{1}{|\bm{M}|}\sum_{w=1}^{V}\sum_{m=1}^{M_w} \tilde{U}_{w,m} \le \frac{N}{|\bm{M}|} \right\}^{-1}\\
&= \Pr\left\{ \tilde{S}_{\bm{M}} \le \frac{N}{|\bm{M}|} \right\}^{-1}.
\end{aligned}
\end{equation}
Notice in the above relationship, $P_0/\tilde{P}_0$ is determined by the probability distribution of $\tilde{S}_M$. Therefore we can use the large-pool limit of $\tilde{S}_{\bm{M}}$ to get the following approximation
\begin{equation} \label{PRatioApprox}
\begin{aligned}
\lim\limits_{|\bm{M}| \to \infty}\frac{P_0}{\tilde{P}_0} 
= \left[1-Q(\frac{N}{|\bm{M}|})\right]^{-1},
\end{aligned}
\end{equation}
where $q(x)$ and $Q(x)$ are respectively the probability density function (PDF) and cumulative tail distribution of $N(\mu,\frac{\sigma^2}{|\bm{M}|})$. With these relationships, we can now approximate the blocking probability:
\begin{equation} \label{pbca}
\begin{aligned}
\lim\limits_{|\bm{M}| \to \infty}P^{bc}
&= \lim\limits_{|\bm{M}| \to \infty}\Pr\left\{ \sum_{w=1}^{V}\sum_{m=1}^{M_w}U_{w,m} = N \right\} \\
&= \lim\limits_{|\bm{M}| \to \infty}\frac{P_0}{\tilde{P}_0}  \Pr\left\{ \sum_{w=1}^{V}\sum_{m=1}^{M_w}\tilde{U}_{w,m} = N \right\} \\
&= \lim\limits_{|\bm{M}| \to \infty}\frac{P_0}{\tilde{P}_0}  \Pr\left\{ \tilde{S}_{\bm{M}} = \frac{N}{|\bm{M}|} \right\}  \\
&= \frac{P_0}{\tilde{P}_0}  \frac{1}{|\bm{M}|}  q\left(\frac{N}{|\bm{M}|}\right),
\end{aligned}
\end{equation}
\begin{equation} \label{pbra}
\begin{aligned}
&\lim\limits_{|\bm{M}| \to \infty}P_v^{br}\\
=& \lim\limits_{|\bm{M}| \to \infty}\Pr\left\{ U_{v,1} = K_v, \sum_{w=1}^{V}\sum_{m=1}^{M_w}U_{w,m} < N \right\} \\
=& \lim\limits_{|\bm{M}| \to \infty}\frac{P_0}{\tilde{P}_0}  \Pr\left\{ \tilde{U}_{v,1} = K_v, \sum_{w=1}^{V}\sum_{m=1}^{M_w}\tilde{U}_{w,m} < N \right\}  \\
=& \lim\limits_{|\bm{M}| \to \infty}\frac{P_0}{\tilde{P}_0}   \Pr\left\{ \tilde{U}_{v,1} = K_v  \right\}  \cdot \\
&\Pr\left\{ \sum_{m=2}^{M_v}\tilde{U}_{v,m} + \sum_{w \neq v}\sum_{m=1}^{M_w}\tilde{U}_{w,m} < N-K_v \right\}\\
=& \lim\limits_{|\bm{M}| \to \infty}\frac{P_0}{\tilde{P}_0}  \tilde{P}_v^{br}  \Pr\left\{ \tilde{S}_{\bm{M}}-\frac{\tilde{U}_{v,1}}{|\bm{M}|-1} < \frac{N-K_v}{|\bm{M}|-1} \right\}\\
=& \frac{P_0}{\tilde{P}_0}   \tilde{P}_v^{br}  \left[1-Q(\frac{N}{|\bm{M}|})\right].\\
\end{aligned}
\end{equation}
Notice the fifth equality of (\ref{pbra}) holds because, as $|\bm{M}| \to \infty$, $N$ should also approach infinity as $N > |\bm{M}|\mu$. Hence
\begin{equation}\label{Q_approx}
\lim\limits_{|\bm{M}| \to \infty}Q(\frac{N-K_v}{|\bm{M}|}) = Q(\frac{N}{|\bm{M}|}).
\end{equation}

Also,$\lim\limits_{|\bm{M}| \to \infty}Q(\frac{N}{|\bm{M}|}) = e^{-\alpha^2/2}$. Therefore, the approximation for the overall session blocking probability of class-$v$ VBSs is
\begin{equation}
\begin{aligned}
\lim\limits_{|\bm{M}| \to \infty}P_v^b
=& \lim\limits_{|\bm{M}| \to \infty}(P^{bc} + P_v^{br}) \\
=& \left[1-Q(\frac{N}{|\bm{M}|})\right]^{-1} \cdot\\
& \left\{\frac{1}{|\bm{M}|}  q\left(\frac{N}{|\bm{M}|}\right)   + \tilde{P}_v^{br}  \left[1-Q(\frac{N}{|\bm{M}|})\right]\right\}\\
=& \frac{\sqrt{|\bm{M}|}}{|\bm{M}|\sqrt{2 \pi \sigma^2}} \frac{e^{-\alpha^2/2}}{1-e^{-\alpha^2/2}} + \tilde{P}_v^{br} \\
=& \frac{1}{\sqrt{2 \pi |\bm{M}| \sigma^2}} \frac{1}{e^{\alpha^2/2}-1} + \tilde{P}_v^{br}.
\end{aligned}
\end{equation}

\section{Proof of Theorem \ref{Theo_lpl}} \label{Proof: lpl}
The first part of our proof is straightforward using (\ref{sm}). Since $\tilde{S}_{\bm{M}}$ will also converge to a normal distributed random variable $N(\mu,\frac{\sigma^2}{|\bm{M}|})$ as $|\bm{M}| \to \infty$, according to the strong law of large numbers:
\begin{equation}
\begin{aligned}
&\Pr\left\{\lim\limits_{|\bm{M}| \to \infty} \eta = \frac{|\bm{M}|\mu}{N} \right\}\\
&= \Pr\left\{\lim\limits_{|\bm{M}| \to \infty} \frac{\sum_{w=1}^{V}\sum_{m=1}^{M_w} \tilde{U}_{w,m}}{N} = \frac{|\bm{M}|\mu}{N} \right\}\\
&= \Pr\left\{\lim\limits_{|\bm{M}| \to \infty} \tilde{S}_{\bm{M}} = \mu\right\} = 1.
\end{aligned}
\end{equation}
Hence $\eta \xrightarrow{\mathrm{a.s.}} \frac{|\bm{M}|\mu}{N}$. Also, it is easy to see that $U_{v,m} \le K_v$ and $\Pr\left\{ U_{v,m} < K_v \right\} >0$. Therefore $\mu_v < K_v$ and \begin{equation}
\begin{aligned}
\mu &= \sum_{w=1}^{V}\mu_w \beta_w \\
 &< \sum_{w=1}^{V}K_w \beta_w
 = \frac{\sum_{w=1}^{V}K_w M_w}{|\bm{M}|}
 = \frac{\bm{M}^T\bm{K}}{|\bm{M}|}\\
 &< \frac{N}{|\bm{M}|}.
\end{aligned}
\end{equation} Thus $\frac{|\bm{M}|\mu }{N} < 1$.




\ifCLASSOPTIONcaptionsoff
  \newpage
\fi

\bibliographystyle{IEEEtran}
\bibliography{IEEEabrv,myref}

\begin{thebibliography}{10}
\providecommand{\url}[1]{#1}
\csname url@samestyle\endcsname
\providecommand{\newblock}{\relax}
\providecommand{\bibinfo}[2]{#2}
\providecommand{\BIBentrySTDinterwordspacing}{\spaceskip=0pt\relax}
\providecommand{\BIBentryALTinterwordstretchfactor}{4}
\providecommand{\BIBentryALTinterwordspacing}{\spaceskip=\fontdimen2\font plus
\BIBentryALTinterwordstretchfactor\fontdimen3\font minus
  \fontdimen4\font\relax}
\providecommand{\BIBforeignlanguage}[2]{{%
\expandafter\ifx\csname l@#1\endcsname\relax
\typeout{** WARNING: IEEEtran.bst: No hyphenation pattern has been}%
\typeout{** loaded for the language `#1'. Using the pattern for}%
\typeout{** the default language instead.}%
\else
\language=\csname l@#1\endcsname
\fi
#2}}
\providecommand{\BIBdecl}{\relax}
\BIBdecl

\bibitem{vni}
``Cisco {V}isual {N}etworking {I}ndex: Global mobile data traffic forecast
  update, 2013-2018,'' 2013.

\bibitem{cran}
\BIBentryALTinterwordspacing
{China Mobile Research Institute}. (2014, June) C-{RAN}: The road towards green
  {RAN} (version 3.0). [Online]. Available:
  \url{http://labs.chinamobile.com/cran/wp-content/uploads/2014/06/20140613-C-RAN-WP-3.0.pdf}
\BIBentrySTDinterwordspacing

\bibitem{wnc}
Y.~Lin, L.~Shao, Z.~Zhu, Q.~Wang, and R.~K. Sabhikhi, ``Wireless network cloud:
  Architecture and system requirements,'' \emph{IBM Journal of Research and
  Development}, vol.~54, no.~1, pp. 4--1, 2010.

\bibitem{ngmn}
{NGMN Alliance}, ``Suggestions on potential solutions to {C-RAN},'' 2013.

\bibitem{concert}
J.~Liu, T.~Zhao, S.~Zhou, Y.~Cheng, and Z.~Niu, ``{CONCERT}: a cloud-based
  architecture for next-generation cellular systems,'' \emph{IEEE Wireless
  Communications}, vol.~21, no.~6, pp. 14--22, December 2014.

\bibitem{2015arXiv151207743S}
O.~{Simeone}, A.~{Maeder}, M.~{Peng}, O.~{Sahin}, and W.~{Yu}, ``Cloud radio
  access network: virtualizing wireless access for dense heterogeneous
  systems,'' \emph{ArXiv e-prints}, Dec. 2015.

\bibitem{nguyen2016sdn}
V.-G. Nguyen, T.-X. Do, and Y.~Kim, ``{SDN} and virtualization-based {LTE}
  mobile network architectures: A comprehensive survey,'' \emph{Wireless
  Personal Communications}, vol.~86, no.~3, pp. 1401--1438, 2016.

\bibitem{colony}
S.~Namba, T.~Matsunaka, T.~Warabino, S.~Kaneko, and Y.~Kishi, ``Colony-{RAN}
  architecture for future cellular network,'' in \emph{Future Network \& Mobile
  Summit (FutureNetw), 2012}.\hskip 1em plus 0.5em minus 0.4em\relax IEEE,
  2012, pp. 1--8.

\bibitem{fluidnet}
K.~Sundaresan, M.~Y. Arslan, S.~Singh, S.~Rangarajan, and S.~V. Krishnamurthy,
  ``{FluidNet}: a flexible cloud-based radio access network for small cells,''
  in \emph{Proceedings of the 19th annual international conference on Mobile
  computing \& networking}.\hskip 1em plus 0.5em minus 0.4em\relax ACM, 2013,
  pp. 99--110.

\bibitem{vbs}
Z.~Zhu, P.~Gupta, Q.~Wang, S.~Kalyanaraman, Y.~Lin, H.~Franke, and S.~Sarangi,
  ``Virtual base station pool: towards a wireless network cloud for radio
  access networks,'' in \emph{Proceedings of the 8th ACM International
  Conference on Computing Frontiers}.\hskip 1em plus 0.5em minus 0.4em\relax
  Ischia, Italy: ACM, 2011, p.~34.

\bibitem{cloudiq}
S.~Bhaumik, S.~P. Chandrabose, M.~K. Jataprolu, G.~Kumar, A.~Muralidhar,
  P.~Polakos, V.~Srinivasan, and T.~Woo, ``Cloud{IQ}: a framework for
  processing base stations in a data center,'' in \emph{Proceedings of the 18th
  annual international conference on Mobile computing and networking}.\hskip
  1em plus 0.5em minus 0.4em\relax Istanbul, Turkey: ACM, 2012, pp. 125--136.

\bibitem{bigstation}
Q.~Yang, X.~Li, H.~Yao, J.~Fang, K.~Tan, W.~Hu, J.~Zhang, and Y.~Zhang,
  ``Big{S}tation: enabling scalable real-time signal processingin large mu-mimo
  systems,'' in \emph{Proceedings of the ACM SIGCOMM 2013 conference}.\hskip
  1em plus 0.5em minus 0.4em\relax Hong Kong, China: ACM, 2013, pp. 399--410,
  2486016.

\bibitem{pran}
\BIBentryALTinterwordspacing
W.~Wu, L.~E. Li, A.~Panda, and S.~Shenker, ``{PRAN}: Programmable radio access
  networks,'' in \emph{Proceedings of the 13th ACM Workshop on Hot Topics in
  Networks}, ser. HotNets-XIII.\hskip 1em plus 0.5em minus 0.4em\relax New
  York, NY, USA: ACM, 2014, pp. 6:1--6:7. [Online]. Available:
  \url{http://doi.acm.org/10.1145/2670518.2673865}
\BIBentrySTDinterwordspacing

\bibitem{sdhcn}
S.~Zhou, T.~Zhao, Z.~Niu, and S.~Zhou, ``Software-defined hyper-cellular
  architecture for green and elastic wireless access,'' \emph{IEEE
  Communications Magazine}, vol.~54, no.~1, pp. 12--19, January 2016.

\bibitem{cpri}
``{CPRI} specification v6.0: Interface specification,'' 2013.

\bibitem{fh}
J.~Liu, S.~Xu, S.~Zhou, and Z.~Niu, ``Redesigning fronthaul for next-generation
  networks: beyond baseband samples and point-to-point links,'' \emph{Wireless
  Communications, IEEE}, vol.~22, no.~5, pp. 90--97, October 2015.

\bibitem{multiplex}
T.~Werthmann, H.~Grob-Lipski, and M.~Proebster, ``Multiplexing gains achieved
  in pools of baseband computation units in 4g cellular networks,'' in
  \emph{Personal Indoor and Mobile Radio Communications (PIMRC), 2013 IEEE 24th
  International Symposium on}, Sept 2013, pp. 3328--3333.

\bibitem{gomez13}
I.~Gomez-Miguelez, V.~Marojevic, and A.~Gelonch, ``Deployment and management of
  {SDR} cloud computing resources: problem definition and fundamental limits,''
  \emph{EURASIP Journal on Wireless Communications and Networking}, vol. 2013,
  no.~1, pp. 1--11, 2013.

\bibitem{vbsmodel}
J.~Liu, S.~Zhou, J.~Gong, Z.~Niu, and S.~Xu, ``On the statistical multiplexing
  gain of virtual base station pools,'' in \emph{2014 IEEE Global
  Communications Conference}, Dec 2014, pp. 2283--2288.

\bibitem{Borst05}
S.~Borst, ``User-level performance of channel-aware scheduling algorithms in
  wireless data networks,'' \emph{Networking, IEEE/ACM Transactions on},
  vol.~13, no.~3, pp. 636--647, June 2005.

\bibitem{bk_jsac}
K.~Son, H.~Kim, Y.~Yi, and B.~Krishnamachari, ``Base station operation and user
  association mechanisms for energy-delay tradeoffs in green cellular
  networks,'' \emph{IEEE Journal on Selected Areas in Communications}, vol.~29,
  no.~8, pp. 1525--1536, September 2011.

\bibitem{ross89}
K.~Ross and D.~H.~K. Tsang, ``The stochastic knapsack problem,'' \emph{IEEE
  Transactions on Communications}, vol.~37, no.~7, pp. 740--747, Jul 1989.

\bibitem{AeinKosovych-37}
J.~M. Aein and O.~S. Kosovych, ``Satellite capacity allocation,''
  \emph{Proceedings of the IEEE}, vol.~65, no.~3, pp. 332--342, 1977.

\bibitem{kaufman81}
J.~Kaufman, ``Blocking in a shared resource environment,'' \emph{IEEE
  Transactions on Communications}, vol.~29, no.~10, pp. 1474--1481, Oct 1981.

\bibitem{PASTA}
R.~W. Wolff, ``Poisson arrivals see time averages,'' \emph{Operations
  Research}, vol.~30, no.~2, pp. 223--231, 1982.

\bibitem{wu13}
J.~Wu, S.~Zhou, and Z.~Niu, ``Traffic-aware base station sleeping control and
  power matching for energy-delay tradeoffs in green cellular networks,''
  \emph{IEEE Transactions on Wireless Communications}, vol.~12, no.~8, pp.
  4196--4209, August 2013.

\end{thebibliography}

%

\begin{IEEEbiography}{Jingchu Liu}
(S'14) received his B.S. degree from Department of Electronic Engineering of Tsinghua University, China, in 2012. He is currently a PhD student at the Department of Electronic Engineering, Tsinghua University. From October 2015 to April 2016, he visited the Autonomous Networks Research Group, Ming Hsieh Department of Electrical Engineering, University of Southern California, CA, USA. His research interests include cloud-based wireless networking, data-driven network management, network data analytics, and green wireless communications.
\end{IEEEbiography}

\begin{IEEEbiography}{Sheng Zhou}
(S'06-M'12) received the B.E. and Ph.D. degrees in electronic engineering from Tsinghua University, Beijing, China, in 2005 and 2011, respectively. From January to June 2010, he was a visiting student at the Wireless System Lab, Department of Electrical Engineering, Stanford University, Stanford, CA, USA. He is currently an Assistant Professor with the Department of Electronic Engineering, Tsinghua University. His research interests include cross-layer design for multiple antenna systems, cooperative transmission in cellular systems, and green wireless communications.	Dr. Zhou coreceived the Best Paper Award at the Asia-Pacific Conference on Communication in 2009 and 2013, the 23th IEEE International Conference on Communication Technology in 2011, and the 25th International Tele-traffic Congress in 2013.
\end{IEEEbiography}

\begin{IEEEbiography}{Jie Gong}
(S'09-M'13) received his B.S. and Ph.D. degrees in Department of Electronic Engineering in Tsinghua University, Beijing, China, in 2008 and 2013, respectively. From July 2012 to January 2013, he visited Institute of Digital Communications, University of Edinburgh, Edinburgh, UK. During 2013-2015, he worked as a postdoctorial scholar in Department of Electronic Engineering in Tsinghua University, Beijing, China. He is currently an associate research fellow in School of Data and Computer Science, Sun Yat-sen University, Guangzhou, Guangdong Province, China. He was a co-recipient of the Best Paper Award from IEEE Communications Society Asia-Pacific Board in 2013. His research interests include Cloud RAN, energy harvesting and green wireless communications.
\end{IEEEbiography}

\begin{IEEEbiography}{Zhisheng Niu}
(M'98-SM'99-F'12) graduated from Beijing Jiaotong University, China, in 1985, and got his M.E. and D.E. degrees from Toyohashi University of Technology, Japan, in 1989 and 1992, respectively.  During 1992-94, he worked for Fujitsu Laboratories Ltd., Japan, and in 1994 joined with Tsinghua University, Beijing, China, where he is now a professor at the Department of Electronic Engineering.  He is also a guest chair professor of Shandong University, China.  His major research interests include queueing theory, traffic engineering, mobile Internet, radio resource management of wireless networks, and green communication and networks.
	
Dr. Niu has been an active volunteer for various academic societies, including Director for Conference Publications (2010-11) and Director for Asia-Pacific Board (2008-09) of IEEE Communication Society, Membership Development Coordinator (2009-10) of IEEE Region 10, Councilor of IEICE-Japan (2009-11), and council member of Chinese Institute of Electronics (2006-11).  He is now a distinguished lecturer (2012-15) and Chair of Emerging Technology Committee (2014-15) of IEEE Communication Society, a distinguished lecturer (2014-16) of IEEE Vehicular Technologies Society, a member of the Fellow Nomination Committee of IEICE Communication Society (2013-14), standing committee member of Chinese Institute of Communications (CIC, 2012-16), and associate editor-in-chief of IEEE/CIC joint publication China Communications.
	
Dr. Niu received the Outstanding Young Researcher Award from Natural Science Foundation of China in 2009 and the Best Paper Award from IEEE Communication Society Asia-Pacific Board in 2013.  He also co-received the Best Paper Awards from the 13th, 15th and 19th Asia-Pacific Conference on Communication (APCC) in 2007, 2009, and 2013, respectively, International Conference on Wireless Communications and Signal Processing (WCSP'13), and the Best Student Paper Award from the 25th International Teletraffic Congress (ITC25).  He is now the Chief Scientist of the National Basic Research Program (so called "973 Project") of China on "Fundamental Research on the Energy and Resource Optimized Hyper-Cellular Mobile Communication System" (2012-2016), which is the first national project on green communications in China.  He is a fellow of both IEEE and IEICE.
\end{IEEEbiography}

\begin{IEEEbiography}{Shugong Xu}
(SM‘06-F’16) received his Ph.D. degree from Huazhong University of Science and Technology in 1996. He is currently the Director of Intel Collaborative Research Institute for Mobile Networking and Computing (ICRI-MNC). Before he took this role, he was Principal Investigator of ICRI-MNC and co-directing the research programs for this new institute after he joined Intel Labs in Sep 2013. Prior to that, he was a research director and principal scientist at the Communication Technologies Laboratory, Huawei Technologies. Among his responsibilities at Huawei, Shugong founded and directed Huawei's green radio research program GREAT. He was also the Chief Scientist and lead for the China National 863 project on End-to-End Energy Efficient Networks. Prior to joining Huawei in 2008, he was with Sharp Laboratories of America as a senior research scientist. Shugong published more than 60 peer-reviewed research papers in top international conferences and journals. One of his most referenced paper has over 1200 Google Scholar citations, in which the findings were among the major triggers for the research and standardization of IEEE 802.11S. Shugong has over 20 US patents granted. Some of these technologies have been adopted in international standards including IEEE 802.11, 3GPP LTE and DLNA. His recent research interests include mobile networking and computing, next generation wireless communication platform, network intelligence and SDN/NFV, etc.
\end{IEEEbiography}

\end{document}